\documentclass[journal]{IEEEtran}
\newlength{\flexwidth}
\usepackage{silence}
\usepackage[table]{xcolor}
\usepackage{cite} 
\usepackage{amsmath,amssymb,amsthm,fixmath}
\usepackage{mathtools}
\usepackage{accents}
\usepackage{xcolor}
\usepackage{siunitx}
\usepackage{cuted}
\usepackage{bm}
\usepackage{multirow,booktabs}
\usepackage{wrapfig}
\usepackage{subcaption}
\usepackage[inline]{enumitem}
\usepackage{optidef}
\usepackage{graphicx}
\usepackage{epstopdf}
\usepackage{lipsum}
\usepackage{amssymb}
\usepackage{diagbox}
\usepackage{adjustbox}
\usepackage{tcolorbox}
\usepackage{float}

\newtheorem{lemma}{Lemma}
\usepackage[normalem]{ulem}

\usepackage{dblfloatfix}
\pdfoutput=1
\usepackage[linesnumbered,ruled,vlined]{algorithm2e}

\usepackage{graphicx}

\usepackage{epstopdf}
\usepackage[textsize=tiny,colorinlistoftodos]{todonotes}
\makeatletter
\define@key{todonotes}{bh}[]{
	\setkeys{todonotes}{author=\textbf{Bin - notes}, color=red!30}}%
\define@key{todonotes}{bh2}[]{
		\setkeys{todonotes}{author=\textbf{Bin - urges}, color=green!30}}%
\define@key{todonotes}{zf}[]{
	\setkeys{todonotes}{author=\textbf{Zexin}, color=blue!30}}%
\makeatother

\usepackage{algpseudocode}

\usepackage[shortcuts,acronym,automake]{glossaries}
\makeglossaries
\newacronym{ue}{UE}{User Equipment}
\newacronym{bs}{BS}{base station}
\newacronym{csi}{CSI}{Channel state information}
\newacronym{b5g}{B5G}{Beyond-Fifth-Generation}
\newacronym{6g}{6G}{Sixth Generation}
\newacronym{ml}{ML}{Machine learning}
\newacronym{sbs}{SBS}{small base station}
\newacronym{mu}{MU}{mobile user}
\newacronym{mbs}{MBS}{macro base station}
\newacronym{mse}{MSE}{Mean Squared Error}
\newacronym{cl}{CL}{centralized learning}
\newacronym{uav}{UAV}{unmanned aerial vehicle}
\newacronym{bme}{BME}{Bayesian Model Ensemble}
\newacronym{iid}{IID}{independent and identically distributed}
\newacronym{raf}{RAF}{robust aggregation function}
\newacronym{sgd}{SGD}{stochastic gradient descend}
\newacronym{cdf}{CDF}{cumulative distribution function}
\newacronym{lid}{LID}{local intrinsic dimensionality}
\newacronym{llpf}{LLPF}{local loss pre-filtering}
\newacronym{mitm}{MITM}{man-in-the-middle}
\newacronym{ae}{AE}{adversary entitie}
\newacronym{tof}{TOF}{time of fly}
\newacronym{rssi}{RSS}{received signal strength}
\newacronym{3d}{3D}{three dimensional}
\newacronym{aoa}{DoA}{Direction of Arrival}
\newacronym{sdp}{SDP}{semi-definite programming}
\newacronym{nlos}{NLOS}{Non-Line-of-Sight}
\newacronym{snr}{SNR}{Signal to Noise Ratio}
\newacronym{crb}{CRB}{Cramer-Rao bound}
\newacronym{lse}{LSE}{least squared estimation}
\newacronym{wlse}{WLSE}{weighted least squared estimation}
\newacronym{gd}{GD}{Gradient descend}
\newacronym{ap}{AP}{Access Points}
\newacronym{crlb}{CRLB}{Cramér-Rao Lower Bound}
\newacronym{tdoa}{TDoA}{Time Difference of Arrival}
\newacronym{sinr}{SINR}{Signal to Interference and Noise Ratio}
\newacronym{los}{LOS}{Line of Sight}
\newacronym{a2g}{A2G}{Air to Ground}
\newacronym{eu}{EU}{European Union}
\newacronym{umiav}{UMi-AV}{Urban Micro–Aerial Vehicle}
\newacronym{3gpp}{3GPP}{3rd Generation Partnership Project}
\newacronym{lae}{LAE}{Low Altitude Economy}
\newacronym{gnss}{GNSS}{Global Navigation Satellite System}
\newacronym{rf}{RF}{Radio Frequency}
\newacronym{gpdr}{GDPR}{General Data Protection Regulation}
\newacronym{5gnr}{5G-NR}{Fifth Generation New Radio}
\newacronym{otdoa}{OTDOA}{Observed Time Difference of Arrival}
\newacronym{prs}{PRS}{Positioning Reference Signals}
\newacronym{gnb}{gNB}{Next Generation Node B}
\newacronym{lmf}{LMF}{Localization Management Function}
\newacronym{amf}{AMF}{Access and Mobility Management Function}
\newacronym{lpp}{LPP}{LTE Positioning Protocol}
\newacronym{nrppa}{NRPPa}{NG-RAN Positioning Protocol A}
\newacronym{prc}{PRC}{Positioning Reference Configuration}
\newacronym{dlotdoa}{DL-OTDOA}{Downlink Observed Time Difference of Arrival}
\newacronym{nas}{NAS}{Non-Access Stratum}
\newacronym{ngc}{NG-C}{Next Generation Control Plane}
\newacronym{tls}{TLS}{Transport Layer Security}
\newacronym{rsrp}{RSRP}{Reference Signal Received Power}
\newacronym{rof}{ROF}{RSS-based optimum finder}
\newacronym{tcv}{TCV}{Triangular Consistency Verification}
\newacronym{sdet}{SDET}{Static Distance-Error Thresholding}
\newacronym{rdef}{RDEF}{Recursive Distance-Error Filtering}
\newacronym{lawn}{LAWN}{Low-altitude Wireless Network}
\newacronym{ls}{LS}{Least Squares}
\newacronym{iq}{IQIA-Net}{In-phase Quadrature Intra-attention Network}
\newacronym{qr}{QR}{Quality report}

\DeclareSIUnit{\belmilliwatt}{Bm}
\DeclareSIUnit{\dBm}{\deci\belmilliwatt}

\usepackage{tcolorbox}
\tcbuselibrary{many}

\usepackage[free-standing-units=true]{siunitx}

\begin{document}
	
	\title{Network-Centric Anomaly Filtering and Spoofer Localization for 5G-NR Localization}
	
	\author{
	 	Zexin~Fang,~\IEEEmembership{Student Member,~IEEE,}
		Bin~Han,~\IEEEmembership{Senior Member,~IEEE,}
        Zhu~Han,~\IEEEmembership{Fellow,~IEEE,}\\
        Yufei~Zhao,~\IEEEmembership{Member,~IEEE,}
        Yong~Liang~Guan,~\IEEEmembership{Senior Member,~IEEE,}
		and
		Hans~D.~Schotten,~\IEEEmembership{Member,~IEEE}
		\thanks{
    Z. Fang, B. Han, and H. D. Schotten are with University of Kaiserslautern-Landau (RPTU), Germany. Z. Han is with University of Huston, United States. Y. L. Guan and Y. Zhao are with Nanyang Technological University, Singapore.
    H. D. Schotten is with the German Research Center for Artificial Intelligence (DFKI), Germany. B. Han (bin.han@rptu.de) is the corresponding author.}
	}
	\bstctlcite{IEEEexample:BSTcontrol}
	
	\maketitle
    \begin{abstract}
        This paper investigates security vulnerabilities and countermeasures for the \gls{3gpp} \gls{5gnr} \gls{tdoa}-based \gls{uav} localization in low-altitude urban environments. We first optimize node selection strategies under \gls{a2g} channel conditions, proving that optimal selection depends on \gls{uav} altitude and deployment density, and propose lightweight \gls{ue}-assisted approaches that reduce overhead while enhancing accuracy. Next, we then expose critical security vulnerabilities by introducing merged-peak spoofing attacks where rogue \glspl{uav} transmit multiple \gls{5gnr} \glspl{prs} that merge with legitimate signals, bypassing existing detection methods. Through theoretical modeling and sensitivity analysis, we quantify how synchronization quality and geometric factors determine spoofing success probability, thereby revealing fundamental weaknesses in current \gls{3gpp} positioning frameworks. To address these vulnerabilities, we design a network-centric anomaly detection framework at the \gls{lmf} using \gls{3gpp}-specified parameters, coupled with recursive gradient descent-based robust localization that filters anomalous data while estimating \gls{uav} position. Our unified framework simultaneously provides robust victim localization and spoofer localization, enabling active attacker attribution beyond passive defense. Extensive simulations validate the effectiveness of our optimization and security mechanisms for \gls{3gpp}-compliant \gls{uav} positioning.
    \end{abstract}

	
	\IEEEpeerreviewmaketitle
	
	\glsresetall

	\section{Introduction}\label{sec:introduction}
    The rapid growth of \gls{lae} has driven demand for reliable communication and precise positioning to support applications such as logistics, inspection, agriculture, and emergency response \cite{LAE2025fang}. \glspl{lawn} have emerged as a key enabler, offering seamless connectivity and situational awareness for aerial platforms in low-altitude airspace. Within \glspl{lawn}, \glspl{uav} function both as users and as flexible network relays, dynamically extending terrestrial coverage in infrastructure-limited regions.

    Among the critical requirements for \gls{lae}, accurate localization underpins navigation, swarm coordination, and safety-critical control, particularly when \gls{gnss} performance degrades. While computer-vision and sensing-based methods achieve high precision, they suffer from computational cost, latency, privacy concerns (e.g., \gls{gpdr} compliance), and limited range. In contrast, \gls{rf}-based localization, particularly \gls{tdoa} approaches, provides a lightweight, infrastructure-compatible alternative \cite{tdoa2022sinha}. \gls{5gnr} features such as dense small-cell deployment, wide bandwidth, and sub-microsecond synchronization enable sub-meter \gls{tdoa} accuracy. \gls{3gpp} has steadily improved \gls{5gnr} localization capabilities, starting with the initial framework introduced in Release 16 \cite{3gpporiframe}. Embedding localization within this standards-based infrastructure thus offers a scalable, low-latency solution for \gls{lawn} applications.

    \gls{5gnr}-based localization has gained significant research momentum. The fusion of \gls{5gnr} and \gls{gnss} has been explored to leverage their complementary strengths, enabling highly accurate, reliable, and continuous localization for low-altitude applications \cite{5gnrintrocampolo,TIMbai5gnr2022}. Beyond serving as positioning targets, \glspl{uav} can function as aerial anchor nodes, significantly enhancing ground user localization by providing dominant \gls{los} links where terrestrial infrastructure is limited \cite{5gnruavwang, 5gnruavliang, Locmagazine2024liang}. Since \glspl{uav} serve as both positioning providers and consumers, precise \gls{uav} localization becomes critical.
    Ground infrastructure-based \gls{3d} localization for \glspl{uav} has attracted attention due to distinct \gls{a2g} channel characteristics. Machine learning-based frameworks using \gls{rssi} from cellular infrastructure \cite{uavlocafifi2021} and \gls{tdoa}-based algorithms exploiting velocity data with convex \gls{ls} optimization \cite{uavmitie2024} have demonstrated \gls{gnss}-independent positioning capabilities. The impact of antenna characteristics has been analyzed through \gls{crb} derivations for various configurations \cite{tdoa2022sinha}, while experimental studies using real flight data have validated altitude-dependent \gls{los} characteristics and their impact on positioning accuracy \cite{uavdickerson2025}. However, practical deployment considerations for ground infrastructure remain underexplored. While \gls{crb} analysis suggests that increasing the number of reference nodes improves localization accuracy, real-world deployment constraints often limit this benefit. Without intelligent link selection, including distant \glspl{gnb} may actually degrade accuracy due to unreliable \gls{tdoa} measurements caused by higher attenuation and reduced \gls{los} probability. This paper addresses this gap by proposing optimal \gls{gnb} selection strategies that account for deployment density, \gls{los} conditions, and \gls{uav} altitude, thereby enhancing the practical viability of \gls{5gnr} localization frameworks for \gls{uav}.

    On the other hand, recent works have revealed 3GPP-compliant vulnerabilities in \gls{5gnr} \gls{tdoa}-based positioning that can manipulate signal timing without disrupting communication. In \cite{spawcspcrosara2025, wcncspzanini2025}, two attack types are investigated: analytical spoofing strategies that alter \gls{prs} propagation times from reference and auxiliary base stations, and full-frame meaconing attacks that intercept, delay, and retransmit entire frames, including \gls{prs}, to bias \gls{tdoa} estimates. Both approaches demonstrate that significant positioning errors can be induced while maintaining normal communication functionality, underscoring the need for robust physical-layer detection mechanisms to safeguard \gls{5gnr} positioning integrity. Corresponding countermeasures have been explored in \cite{wclspspanos2025}, this work proposes a  positioning authentication scheme that secures the \gls{prs} by embedding a hash-based message authentication code (HMAC) into its unused resource elements. A similarity threshold ensures reliable verification under low \gls{snr} and common physical-layer impairments.
   The work in \cite{iccwspFoca2025,twcspfoca2025} focuses on modeling and detecting \gls{prs} spoofing (replay) attacks targeting positioning systems. The authors consider scenarios where an attacker re-transmits delayed PRS replicas, creating separate spoofed and legitimate correlation peaks that bias timing-based measurements. They develop a mathematical threat model and propose detection methods based on cross-correlation analysis and Gaussian Mixture Models (GMMs). The work in \cite{infocom2023Gao, Secure5gnrgao} exposes a critical \gls{5gnr} positioning vulnerability by selectively tampering specific \gls{prs} resource elements to bias location measurements, evading \gls{3gpp} R18 defenses. To defend against attack, the authors propose \gls{iq}, a physical-layer deep learning method that extracts hardware-specific \gls{rf} fingerprints from I/Q samples, converts them into “IQ-images,” and uses an attention mechanism to detect spoofed signals. 
    
    Regardless of attack vector, spoofing aims to create spurious correlation peaks corrupting \gls{tdoa} measurements. Existing countermeasures face practical limitations: \cite{wclspspanos2025} requires protocol modifications and processing overhead impractical for energy-constrained \glspl{uav}, \gls{ue}-based approaches \cite{iccwspFoca2025, twcspfoca2025} burden resource-limited devices, while \gls{gnb}-side methods \cite{infocom2023Gao, Secure5gnrgao} process high-volume I/Q samples through deep neural networks, introducing significant latency when monitoring multiple links. This reveals two {\em critical gaps}. Detection must be lightweight and network-centric to avoid burdening \glspl{ue}. Existing work oversimplifies adversaries by focusing on defense without analyzing attack constraints. Successful \gls{tdoa} spoofing faces fundamental synchronization and power limitations that, when characterized, inform efficient countermeasures. Moreover, current defenses assume spoofed peaks are distinct enough for detection. However, sophisticated adversaries can transmit multiple lower-power pulses that merge with legitimate signals, creating inseparable biases evading existing methods \cite{iccwspFoca2025, twcspfoca2025}. Understanding these constraints is essential for robust defense mechanisms.

    This paper is a rigorous extension of \cite{fang2025lightweight}. While the previous work primarily focused on optimal node selection under \gls{a2g} channels for sensor networks, we shift focus to \gls{5gnr} positioning with \gls{3gpp} compliance considerations. Our key contributions beyond the previous work while addressing above mentioned research gaps include:
    \begin{itemize}
     \item We investigate \gls{ue}-assisted and \gls{lmf}-coordinated node selection strategies for \gls{5gnr}-based \gls{uav} localization, which are undefined in current \gls{3gpp} specifications, revealing that weighted positioning enhances performance, especially under suboptimal node selection.

    \item We introduce merged-peak spoofing attacks that bypass existing detection frameworks by modeling rogue \glspl{uav} as network-integrated attackers. Through theoretical modeling and sensitivity analysis, we quantify how synchronization quality and geometric factors determine attack effectiveness, exposing fundamental vulnerabilities in current schemes.

   \item We design a lightweight, network-centric anomaly detection framework at the \gls{lmf} using existing \gls{3gpp}-specified parameters, coupled with a recursive gradient descent-based algorithm that simultaneously filters anomalies and estimates \gls{uav} position.

   \item Beyond detection, we integrate spoofer localization into the \gls{lmf} using filtered anomaly data, creating a unified framework for robust victim positioning and simultaneous attacker localization—an integration unexplored in existing literature. We analyze the interplay between detection and localization performance.
     \end{itemize}

    The remainder of this paper is organized as follows. Section~\ref{sec:pre} introduces the \gls{5gnr} localization framework and investigates the \gls{a2g} channel model to provide a foundation for subsequent analysis. Section~\ref{sec:sys_model} derives the \gls{crb} for \gls{tdoa} measurements, formulates the localization optimization objective, and proposes a lightweight \gls{ue}-assisted node selection algorithm. Section~\ref{sec:prsthreatmodel} presents the spoofing attack model and our network-centric defense framework. Section~\ref{sec:simu} evaluates the proposed approaches through simulation, and finally Section~\ref{sec:conclu} concludes the paper.

    \section{Preliminaries}\label{sec:pre}
    \subsection{\gls{3gpp} localization framework for otdoa-based localization}
    We consider a \gls{5gnr} localization system compliant with the \gls{3gpp} architecture, where the location of a \gls{ue} (e.g., a \gls{uav}) is estimated using \gls{otdoa}-based localization. In this method, \glspl{gnb} transmit \gls{prs}, and are required to be time-synchronized to ensure accurate multilateration. The localization procedure is coordinated by the \gls{lmf}, which is responsible for computing the \gls{ue}'s location and accounts for inter-cell synchronization errors through calibration.

     \begin{figure}[!ht]
    \centering
        \centering
        \includegraphics[width=0.99\linewidth]{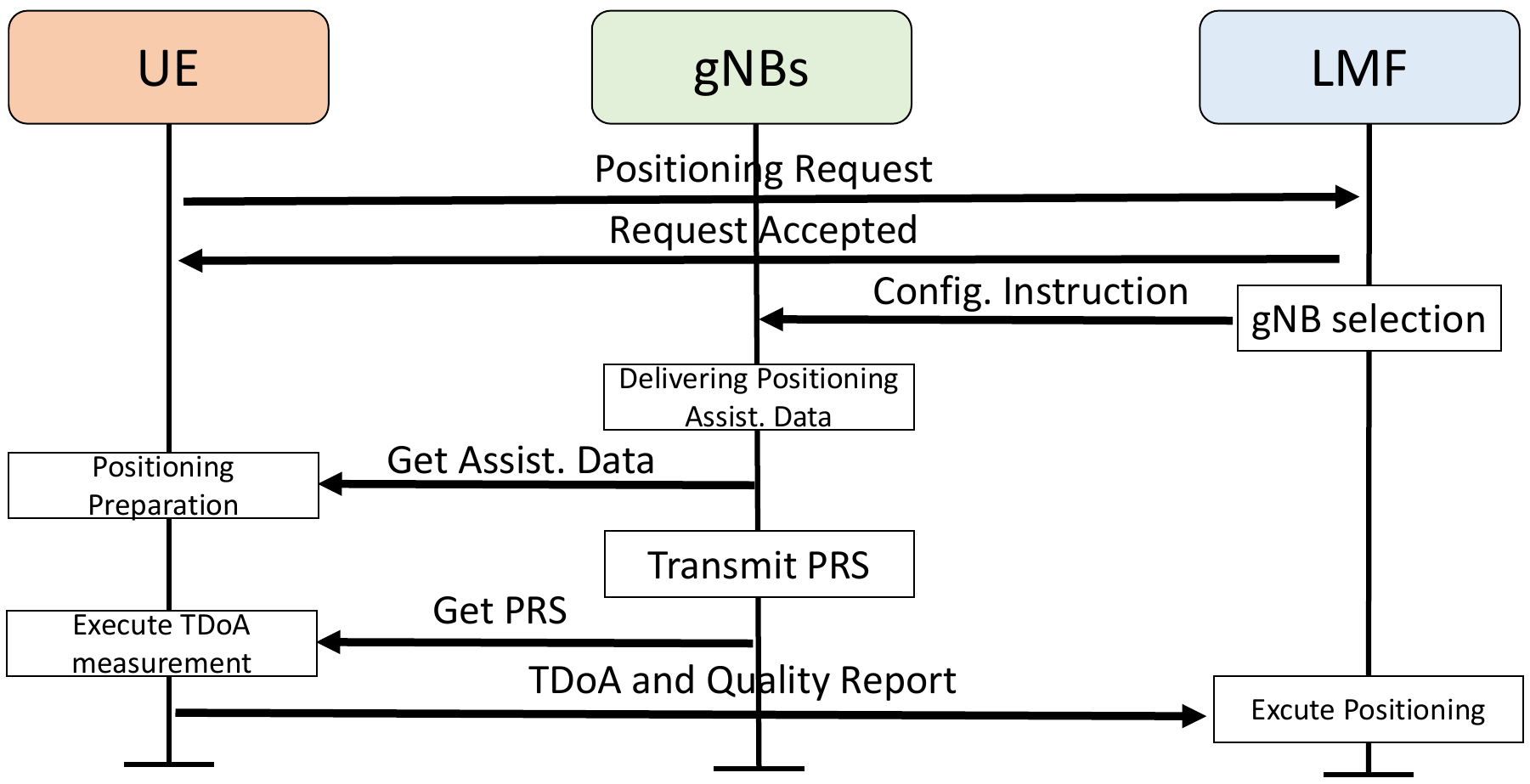}
    \caption{\gls{5gnr} \gls{ue} localization operation flow}
    \label{fig:5gnrflow}
   \end{figure}

  The \gls{amf} handles connection management, access control, and mobility for the \gls{ue}, serving as the signaling anchor between \gls{ue} and core network. While not directly involved in localization, it tunnels \gls{lpp} messages between the \gls{lmf} and \gls{ue}. Key localization protocols and signals for \gls{otdoa}-based positioning are: \begin{enumerate*}[label=\emph{\roman*)}]
  \item\textbf{LPP}:
  \gls{lpp} is the primary signaling protocol used between the \gls{lmf} and the \gls{ue} to facilitate localization procedures. It is standardized in \cite{3gpp.37.355} for \gls{5gnr}. The \gls{lpp} enables the exchange of localization-related messages between the \gls{ue} and \gls{lmf}. First, the \gls{ue} can initiate a localization request. The \gls{lmf} then sends a measurement configuration message, prompting the \gls{ue} to receive localization \gls{prs} from \glspl{gnb} by \gls{prc}. The \gls{ue} subsequently performs measurements and sends a measurement report back to the \gls{lmf}. 
  \item \textbf{NRPPa}: the \gls{lmf} uses the \gls{nrppa} protocol to coordinate \gls{prs} transmission by the \glspl{gnb} \cite{3gpp.38.355}. Through \gls{nrppa}   signaling, the \gls{lmf} sends measurement request messages to instruct one or more \glspl{gnb} to configure and transmit \gls{prs}, specifying parameters such as timing, duration, beam configuration. 
  \item\textbf{PRS}:
  \gls{prs}, defined by the \gls{prc}, are transmitted by \glspl{gnb}. The \gls{ue} receives \gls{prs}, measures \gls{tdoa}, and reports results with \gls{qr} to the \gls{lmf}. \gls{prs} can be time-multiplexed or frequency-multiplexed across \glspl{gnb}. While time-multiplexing is simpler to implement, frequency-multiplexing enables simultaneous transmissions, significantly reducing latency for localizing \gls{uav}.
  \end{enumerate*}
  
  Currently, \gls{dlotdoa} is the standardized \gls{5gnr} localization method, where the \gls{ue} passively measures \gls{tdoa} and reports to the \gls{lmf}, which computes the position estimate (Fig.~\ref{fig:5gnrflow}). This approach reduces signaling overhead and offloads computational burden from energy-constrained \glspl{ue}. Therefore, we consider frequency-multiplexed \gls{prs}-enabled \gls{dlotdoa} as specified in \gls{3gpp} Release 17 for high-accuracy positioning \cite{3enhancefre}. 
  \subsection{\gls{3gpp} A2G channel model}\label{subsec:a2gchannel}
    \gls{3gpp} conducted comprehensive studies on \gls{a2g} channel characteristics in \cite{3gpp-tr36.777}. The \gls{umiav} channel model characterizes \gls{los} probability as: 
     \begin{equation}\label{eq:problos}
        P_\text{los}=  \begin{cases}
       1 ,\quad d_{\text{2D}} \leq d_1 ;\\
       \left(1-\frac{d_1}{d_{\text{2D}}}\right)\exp{\left(\frac{-d_{\text{2D}}}{p_1}\right)}+ \frac{d_1}{d_{\text{2D}}} ,\quad d_{\text{2D}} > d_1.
        \end{cases}
    \end{equation}
    Here, $d_\text{2D}$ denotes the horizontal separation between the aerial platform and terrestrial base station, while $h$ indicates altitude. The corresponding parameters $d_1$ and $p_1$ are given by:
    \begin{equation}\label{eq:d1p1}\begin{split}
        d_1 =& \max(294.05\log_{10}(h) - 432.94, 18);\\
        p_1 =& 233.98\log_{10}(h) - 0.95.\\
    \end{split}
    \end{equation}
    By incorporating both \gls{los} and \gls{nlos} propagation scenarios, we obtain the composite average path loss: 
    \begin{equation}\begin{split}
        \eta =& (4.32 - 0.76\log_{10}(h))(1-P_\text{los}) \\
        &+  (2.225 - 0.05\log_{10}(h))P_\text{los}.
    \end{split}
   \end{equation}
    Above formulation enables spatial categorization based on channel dominance. Region A1 occurs when $d_{\text{2D}} \leq d_1$, where \gls{los} propagation prevails. Beyond this threshold $d_{\text{2D}} > d_1$, region A2 emerges with probabilistic \gls{los}/\gls{nlos} behavior governed by range and elevation parameters. From this channel model, we derived analytical expressions for the average path loss exponent $\eta$ in Eq.~(\ref{eq:eta1}), along with its first-order and second-order derivatives over $d_\text{2d}$ in following equations.
   \begin{align}\label{eq:eta1}
   \eta =
   \begin{cases}\begin{split}
     2.225& - 0.05\log_{10}h,\quad \text{if in A1},\\
    (4.32 &- 0.76\log_{10}(h))(1-P_\text{los}) \\
        &+  (2.225 - 0.05\log_{10}(h))P_\text{los},\quad \text{if in A2}.
     \end{split}
   \end{cases}
   \end{align}
   \begin{align}\label{eq:etaderia}
   \eta' =
   \begin{cases}\begin{split}
     &0,\quad \text{if in A1},\\
     &\bigg(\underbrace{0.71\log_{10}(h) - 2.07}_{<0}\bigg) \Bigg(\Bigg[\underbrace{-\frac{d_1}{d^2_{\text{2D}}}- 
      \frac{1 - \frac{d_1}{d_{\text{2D}}}}{p_1} }_{<0}\Bigg]\\
     & \quad\quad\underbrace{\exp\left(-\frac{d_{\text{2D}}}{p_1}\right)}_{>0} - \underbrace{\frac{d_1}{d^2_{\text{2D}}}}_{>0}\Bigg)
   \quad \text{if in A2}.
     \end{split}
   \end{cases}
   \end{align}
   
    \begin{align}\label{eq:etasecderia}
   \eta'' =
   \begin{cases}\begin{split}
    &0,\quad \text{if in A1},\\  
    &\bigg(\underbrace{0.71\log_{10}(h) - 2.07}_{<0}\bigg)\bigg(\bigg[\underbrace{\frac{2d_1}{d_{\text{2D}}^3} + \frac{d_{\text{2D}}-d_1}{p^2_1 d_{\text{2D}}}\bigg]}_{>0} \\
    &\quad\quad\underbrace{\exp\left(-\frac{d_{\text{2D}}}{p_1}\right)}_{>0} + \underbrace{\frac{2d_1}{d_{\text{2D}}^3}}_{>0}\bigg) \quad \text{if in A2}.   
     \end{split}
   \end{cases}
   \end{align}
   
   Operational altitude constraints dictate \gls{uav} flight parameters: minimum elevation requirements prevent urban collision hazards, while maximum ceiling limits ensure regulatory compliance \cite{bmvd2021drones}. We establish the boundary: $h \in [20,120]$. Within this boundary, the inequality $0.71\log_{10}(h) - 2.07 < 0$ is satisfied, ensuring $\eta' \geq 0$ and $\eta'' \leq 0$ simultaneously. This mathematical relationship confirms that $\eta$ exhibits monotonic growth with progressively decreasing incremental rates. Fig.~\ref{fig:a2gchannel} illustrates the corresponding averaged $\eta$ values across the operational envelope. The \gls{a2g} framework from \cite{3gpp-tr36.777} maintains validity exclusively above $\SI{22.5}{\meter}$ altitude. Propagation models from \cite{3gpp.38.901} are specified for lower elevations. Nevertheless, we adopt a simplified approach by extrapolating the \gls{a2g} characterization to $\SI{20}{\meter}$ altitude for analytical consistency.
  \begin{figure}[!t]
    \centering
        \centering
        \includegraphics[width=0.62\linewidth]{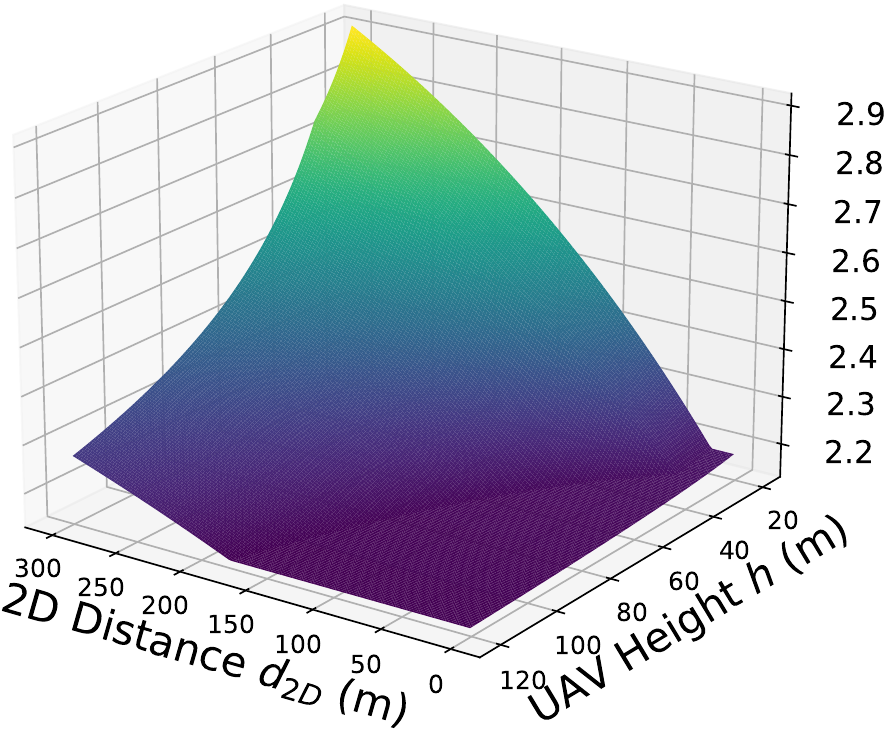}
    \caption{Path loss component $\eta$ with respect to $d_{\text{2D}}$ and $h$ (A1 is the black plane and A2 is the curved surface)}
    \label{fig:a2gchannel}
\end{figure}
	\section{Optimal node selection for localization}\label{sec:sys_model}
    \subsection{Objective equivalence}
    The \gls{crb} for \gls{3d} localization can be expressed by the modeled distance error $\sigma_{M}$ and the number of reference nodes $N$, with $\sigma^2_{\theta} \ge \frac{6\sigma^2_{M}}{N}$. Which assumes uniform distribution of all reference nodes around the \gls{uav}, with numerical validation provided in \cite{GD2025fang}. While $\sigma_{M}$ encompasses both position error variance $\sigma_p$ and distance error variance $\sigma_d$, we can simplify the analysis by treating static references as reliable. Despite neglecting $\sigma_p$, $\sigma_{M}$ remains mathematically intractable. Therefore, we adopt a tractable approximation:
    \begin{equation}\label{eq:loccrb2}
    {\sigma}^2_{\theta} \propto \frac{1}{N}\big(\frac{\sum_{n=1}^N\sigma_{d}^n}{N}\big)^2,
   \end{equation}
    where ${\sigma}_{n,d}$ represents the variance of distance estimation errors for each link. Research findings in \cite{Venus2020ToA, Hechen2024ToA} establish that dense multipath environments impose a fundamental limit on \gls{tdoa} measurement precision through $\sigma_d^2 \geq J_{\text{T}}^{-1}$, with $J_{\text{T}}$ denoting the Fisher information for \gls{tdoa} estimation:
    \begin{equation}\label{eq:toacrb}
    J_{\text{T}} = {2 c^{-1} 4 \pi^2 \text{SINR}\gamma \beta^2 \sin^2(\phi)}.
    \end{equation}
    In this expression, $c$ and $\beta$ represents the light velocity and bandwidth, while $\text{SINR}$ captures the \gls{sinr} affected predominantly predominantly by multipath and cross-interference. Given that orthogonal \glspl{prs} from different \glspl{gnb} are mandated by \gls{3gpp} and multipath interference in \gls{a2g} is generally minimal, we therefore adopt the model $\text{SINR} \approx \text{SNR}$. The terms $\gamma$ and $\sin^2(\phi)$ capture whitening effectiveness and information degradation from path loss parameter uncertainty, respectively. Studies in \cite{WGKlaus2016} reveal that although both quantities depend on bandwidth $\beta$, the sensitivity of $\sin^2(\phi)$ proves substantially lower than $\gamma$. This observation allows us to identify $\text{SNR}$, $\beta$, and $\gamma$ as the dominant factors governing $J_{\text{T}}$. The whitening effectiveness $\gamma$ follows:

    \begin{equation}\nonumber
    \gamma = \frac{\text{SNR}_\text{w}}{\text{SNR}} = \frac{P_{\text{wpre}}}{P_{\text{wpost}}},
    \end{equation}
    \begin{equation}\nonumber
    P_{\text{wpre}} = \int_{-\frac{\beta}{2}}^{\frac{\beta}{2}} s_n(f) df; \quad P_{\text{wpost}} = \text{S}_n^w\beta.
    \end{equation}
    Here, $P_{\text{wpre}}$ and $P_{\text{wpost}}$ represent pre- and post-whitening noise power levels. The whitening operation flattens the noise spectrum to $\text{S}_n^w$. Adopting a power-law noise characterization $s_n(f) = C f^{-\alpha}$ leads to:
     \begin{equation}\nonumber\begin{split}
       P_{\text{wpre}} &= \int_{-\frac{\beta}{2}}^{\frac{\beta}{2}} C f^{-\alpha} df =2\alpha C\ln(\frac{\beta}{2}).
     \end{split} 
     \end{equation}
    Since $\alpha = 1$ is commonly observed in real-world systems, 
    \begin{equation}\nonumber
     \gamma = \frac{P_{\text{wpre}}}{P_{\text{wpost}}} = \frac{2C\ln(\frac{\beta}{2})}{\text{S}_n^w\beta}.
    \end{equation}
    For practical bandwidth values where $\ln(\beta) \gg \ln(2)$, this reduces to the approximation $\gamma \propto \ln(\beta)\beta^{-1}$. \gls{a2g} propagation genrally follows Rician statistics. Urban scenarios typically exhibit stable $K$ factors with altitude sensitivity. Integrating the whitening gain into Eq.~(\ref{eq:toacrb}), we establish that distance error variance depends primarily on two mechanisms through: $\sigma^2_{n,d} \propto(\text{SNR}_n\beta_n\ln\beta_n)^{-1}$. Practical systems typically maintain proportional relationships between transmission power and allocated bandwidth, mirrored in noise characteristics:
    \begin{equation}\label{eq:transp}
    P_n = \psi \beta_n, \quad N_n = N_0 \beta_n,
    \end{equation}
    where $\psi$ and $N_0$ denote transmission and noise power spectral densities. Due to the narrow sub-band relative to the carrier frequency, inter-channel frequency differences can be neglected. These considerations yield $\text{SNR}_n \propto d_{\text{3D},n}^{-\eta_n}$. Combined with $\sigma^2_{n} \propto d_{\text{3D},n}^{\eta_n}(\beta_n\ln\beta_n)^{-1}$, incorporation into Eq.~(\ref{eq:loccrb2}) produces:
    \begin{equation}\label{eq:locobj}
    {\sigma}^2_{\theta} \propto {N}^{-3}\left(\sum_{n=1}^N d_{\text{3D},n}^{\frac{1}{2}\eta_n}(\beta_n\ln\beta_n)^{-\frac{1}{2}}\right)^2.
   \end{equation}
  
   \subsection{Reference deployment}\label{subsec:deploy}
   When an \gls{uav} is flying above a dense urban area, it is reasonable to assume that \glspl{gnb} are available in multiple directions due to the high density of infrastructure. To facilitate optimization and analytical tractability, we model the horizontal distance to the $N_\text{th}$ nearest reference node, denoted as $d_{\text{2D}}(N)$, using a monotonically increasing function. For a hexagonal grid topology for reference deployment, 
  $d_{\text{2D}}(N)$ is:
    \begin{equation}\label{eq:Dnaps}
    d_{\text{2D}} (N)
    \begin{cases}
    =\Delta, 
    \quad \text{if } N = 1, \\
    \in [R - \Delta, R + \Delta],
    \quad \text{if } 1 < N \leq 7, \\
    \vdots  \\
    \in [kR - \Delta, kR + \Delta], \\
    \quad \text{if } 3k^2 - 3k + 1 < N \leq 3k^2 + 3k + 1,
   \end{cases}
   \end{equation}
    where $\Delta$ represents the distance to the nearest reference node, $R$ defines the reference node coverage radius, with $\Delta\in [0,R/2]$. The parameter $k$ indicates the hexagonal grid layer index. Although $d_{\text{2D}}(N)$ exhibits piecewise behavior, its segmental derivatives can be characterized as:
    \begin{equation}\label{eq:Dnaps2}
    d'_{\text{2D}}(N)  
    \begin{cases}
    =\Delta, & N = 1 ,\\
    \approx\frac{\Delta}{3}, &1 < N \leq 7,\\
    \vdots\\
    \approx\frac{\Delta}{3k}, &3k^2 - 3k + 1 < N \leq 3k^2 + 3k + 1.
    \end{cases}
    \end{equation}
    Statically, we take $d'_{\text{2D}}(N)>0$ while $d''_{\text{2D}}(N) \approx 0$ across most operational scenarios. 

   \subsection{Localization optimization}
   Multi-\gls{uav} deployments utilizing shared terrestrial infrastructure typically employ predetermined sub-band allocations to facilitate spectrum partitioning and mitigate co-channel interference. Such fixed bandwidth assignment strategies streamline radio resource management at \glspl{gnb}. Under this operational framework, we examine localization accuracy assuming constant bandwidth allocation $\beta_n$ for the $N_\text{th}$ \gls{uav}. This constraint allows Eq.~(\ref{eq:locobj}) to be reduced to:

   \begin{equation}\label{eq:locobj1}
    {\sigma}^2_{\theta} \propto {N}^{-3}\left(\sum_{n=1}^N d_{\text{3D},n}^{\frac{1}{2}\eta_n} \right)^2.
   \end{equation}
   We then formulate the following: 
    \begin{equation}\label{eq:locobj2}
   \sum_{n=1}^N d_{\text{3D},n}^{\frac{1}{2}\eta_n} = \sum_{n=1}^N (d_{\text{2D},n}^2 + h^2)^{\frac{1}{4}\eta(d_{\text{2D},n}, h)},
   \end{equation}
   \begin{equation}\label{eq:capphid2dh}
     \Phi(N) =  \sum_{n=1}^N (d_{\text{2D},n}^2 + h^2)^{\frac{1}{4}\eta(d_{\text{2D},n}, h)} ,
   \end{equation}
   \begin{equation}\label{eq:phid2dh}
     \phi(N) = (d_{\text{2D},n}^2 + h^2)^{\frac{1}{4}\eta(d_{\text{2D},n}, h)}.
   \end{equation}
   Building upon Eq.~(\ref{eq:locobj1}), the corresponding optimization framework becomes:
   \begin{equation} \label{eq:opt_problem}
  \begin{aligned}
   \underset{{N}}{\text{min}} \quad &  F_\theta(N) = \Phi^2(N)N^{-3} \\
   \text{s. t.} :\quad & N \in \mathbb{Z}_{>1}, \\
                       & h \in \mathbb{R}^{+},\\
                       & d_{\text{2D},n} \in \mathbb{R}^{+}, \quad n \in \mathcal{N}.
   \end{aligned}
   \end{equation}
  From our prior work \cite{fang2025lightweight}, $F_\theta(N)$ admits an optimal solution $N_\text{opt}$ under the statistical \gls{a2g} channel model and hexagonal deployment (Subsecs.~\ref{subsec:a2gchannel}, \ref{subsec:deploy}). Due to space limitations, the complete proof is provided in \cite{fang2025lightweight}.
  
  With known $d_{\text{2D}} (N)$ and latitude $h$ precisely, a general optimal $N_\text{opt}$ can be obtained. However, that is practically infeasible, meanhwile, a general optimal $N_\text{opt}$ cannot guarantee optimized performance for each individual case where \gls{los} conditions vary. Therefore, an adaptive yet lightweight approach for node selection is urgently needed in the current \gls{3gpp} localization framework.
   
   \subsection{Lightweight node selection method}\label{sec:method}
   Although \gls{3gpp} defines procedures for \gls{dlotdoa}, node selection strategies are not specified. The \gls{lmf} can select the nearest \glspl{gnb} using approximate \gls{gnss} position information—termed \gls{lmf}-coordinated node selection. Alternatively, the \gls{ue} can select nodes during handover based on \gls{rsrp} and inform the \gls{lmf} via \gls{lpp}—termed \gls{ue}-assisted node selection. We analyze both strategies below.
   \subsubsection{\gls{lmf} coordinated node selection} First, we can rewrite Eq.~\ref{eq:phid2dh} in the form as $\phi(N) = (d_{\text{3D},n})^{\frac{1}{2}\eta_{n}}$, and $\eta_n$ subjects to a two-point distribution with probability determined by $d_{\text{3D},n}$:
    \begin{equation}\label{eq:Berno}
        \eta_n =\begin{cases}4.32 - 0.76 \log(h), \quad \text{w.p. } 1-P_\text{los},\\
        0.225-0.05\log(h), \quad \text{w.p. } P_\text{los}.\end{cases}
    \end{equation}
   This approach sorts the nodes based on inaccurate distance estimates. While ignoring the \gls{gnss} error, we formulate the selection rank sequence with a form of $\phi(N)$,
   \begin{equation}\label{eq:rlselection}
        \textbf{r}_\text{L} = \operatorname{sort}(\textbf{d}_\text{3D})^{\odot\frac{1}{2}\bm{\eta}},
    \end{equation}
   where $\textbf{d}_\text{3D}$ and $\bm{\eta}$ denote the distance vector and path-loss exponent vector, respectively. $\textbf{r}_\text{L}$ is stochastically monotonically increasing, with:
   \begin{equation}\label{eq:rlmono}
   \mathbf{r}_\mathrm{L}^{(1)} \leq_{\mathrm{st}} \mathbf{r}_\mathrm{L}^{(2)} \leq_{\mathrm{st}} \cdots \leq_{\mathrm{st}} \mathbf{r}_\mathrm{L}^{(K)}.
   \end{equation}
   \subsubsection{\gls{ue} assisted node selection} Taking into account that $\textbf{d}_\text{3D}$ and $\bm{\eta}$ are identical, the rank sequence based on \gls{rssi} can be described by, 
   
    \begin{equation}\label{eq:ruselection}
        \textbf{r}_\text{U} = \operatorname{sort}(\textbf{d}_\text{3D}^{\odot\frac{1}{2}\bm{\eta}}).
    \end{equation}
   For the first element of $\textbf{r}_\text{L}$, since $\eta_n$ subjects to a two-point distribution, we have,
  \begin{equation}\label{eq:rlprob}
        P(\textbf{r}_\text{L}^{(1)}\leq \overline{\textbf{r}}_\text{L}^{(1)}) = P_\text{LOS}.
    \end{equation}
    where $\overline{\textbf{r}}_\text{L}^{(1)}$ is the averaged value resulting from variation of $\eta_n$. Meanwhile for $\textbf{r}_\text{U}$, we have,
   \begin{equation}\label{eq:ruprob}\begin{split}
    &\forall \textbf{r}_\text{L}^{(k)}, P(\textbf{r}_\text{L}^{(k)}\leq \overline{\textbf{r}}_\text{L}^{(1)}) > 0;\\
       &P(\textbf{r}_\text{U}^{(1)}\leq \overline{\textbf{r}}_\text{L}^{(1)}) = P(\textbf{r}_\text{L}^{(1)}\leq \overline{\textbf{r}}_\text{L}^{(1)}) + P(\textbf{r}_\text{L}^{(2)}\leq \overline{\textbf{r}}_\text{L}^{(1)})+...
   \end{split}
    \end{equation}
    Applying the definition of usual stochastic order \cite{shaked2007stochastic}, for random $\overline{\mathbf{r}}_\mathrm{L}^{(1)} \in \mathbb{R}$, it has:
    \begin{equation}\label{eq:stoor}
    \mathbf{r}_\mathrm{U}^{(1)} \leq_{\mathrm{st}} \mathbf{r}_\mathrm{L}^{(1)} \iff 
    P\big(\mathbf{r}_\mathrm{U}^{(1)} \leq \overline{\mathbf{r}}_\mathrm{L}^{(1)}\big) 
    \geq 
    P\big(\mathbf{r}_\mathrm{L}^{(1)} \leq \overline{\mathbf{r}}_\mathrm{L}^{(1)}\big)
    .
    \end{equation}
    $\mathbf{r}_\mathrm{U}^{(1)}$ is stochastically smaller than $\mathbf{r}_\mathrm{L}^{(1)}$. Similarly, for the last element, $\mathbf{r}_\mathrm{U}^{(K)} \geq_{\mathrm{st}} \mathbf{r}_\mathrm{L}^{(K)}$ holds. Given the boundary conditions and the stochastic monotonicity of both sequences, there exists an index $k_\text{e}$ where the stochastic ordering relationship transitions, i.e., $r_\mathrm{U}^{(k_\text{e})} \approx_{\mathrm{st}} r_\mathrm{L}^{(k_\text{e})}$. Therefore, the error terms from both selection strategies, $\phi_\text{L}(N)$ and $\phi_\text{U}(N)$, follow the same trend:  
    \begin{align}
    \phi_\mathrm{U}(N) &\leq_\mathrm{st}  \phi_\mathrm{L}(N), \quad \text{for } 1 \leq N \leq k_\mathrm{e}, \nonumber \\
    \phi_\mathrm{U}(N) &\geq_\mathrm{st}  \phi_\mathrm{L}(N), \quad \text{for } k_\mathrm{e} \leq N \leq K; \label{eq:phi_bound}
    \end{align}
    Consequently, there exists a point $K_\text{e}$ (where $K_\text{e} > k_\text{e}$) such that for the cumulative error functions $\Phi_\text{L}(N)$ and $\Phi_\text{U}(N)$: 
    \begin{align}
    \Phi_\mathrm{U}(N) &\leq_\mathrm{st} \Phi_\mathrm{L}(N), \quad \text{for } 1 \leq N \leq  K_\mathrm{e}, \nonumber \\
    \Phi_\mathrm{U}(N) &\geq_\mathrm{st} \Phi_\mathrm{L}(N), \quad \text{for } K_\mathrm{e} \leq N \leq K. \label{eq:Phi_bound}
    \end{align}
    Referring Eq.~(\ref{eq:opt_problem}), we can conclude that, in case of a resource-limited scenario, where $N_\text{max} \leq K_\text{e}$ \gls{ue} assisted node selection performs better. And in the a scenario where resource is sufficient  \gls{lmf} coordinated node selection performs better. 
    
    In practice, total bandwidth supporting \gls{dlotdoa} is limited, and resource constraints extend beyond just the \gls{uav}. We focus on solving $N_\text{opt}$ via \gls{ue}-assisted node selection, which offers: \begin{enumerate*}[label=\emph{\roman*)}]
   \item extremely lightweight operation with no processing overhead;
   \item optimal performance in low-$N$ regimes;
   \item sensitivity to \gls{los} probability mismatches, enabling adaptive compensation
   \end{enumerate*}. Algorithm~\ref{alg:rbof} presents the detailed implementation. Lines 12-17 compensate for path-loss variation. In dense urban areas, $\mathbf{T}_2$ (${T}_{2,n} = 0$ is the break point where the localization error increases with $N$, details can be found in \cite{fang2025lightweight}) receives a positive offset, forcing $N_\text{opt}$ smaller, so we apply negative feedback to dynamically enlarge $N_\text{opt}$. Following optimal node selection, \gls{tdoa} measurements from selected nodes feed into the gradient-based localization algorithm from \cite{GD2025fang} (Alg.~\ref{alg:MAGD}). This algorithm outperforms conventional methods, particularly when reference nodes have limited altitude variation. It fuses prior \gls{uav} position from previous estimates or \gls{gnss} to improve accuracy even when $N<3$, unlike conventional \gls{3d} methods requiring at least three reference nodes. 
    \begin{algorithm}[!t]
    \caption{\gls{rof}}
    \label{alg:rbof}
    \scriptsize
    \DontPrintSemicolon
    Input: estimated distances based on \gls{rssi} $d^R_{n}$; imperfect altitude $h$; averaged path loss component consolidated as a table $\overline\eta(d_\text{2D}, h)$; maximum reference numbers $N_m$ \\
    \SetKwProg{Fn}{Function}{ :}{end}
    \Fn{}{ 
        sort $d^R_{n}$; \\
        \For {$n = 1:N_m$ }{
        $d_{\text{2D},n} = \sqrt{(d^R_{n})^2 - h^2}$\\
        find $\overline\eta$ regarding $d_{\text{2D},n}$ and $h$\\
        compute $\phi_n$ referring Eq.~(\ref{eq:phid2dh})
        }
        compute $\Phi_n$ referring Eq.~(\ref{eq:capphid2dh})\\
        compute $T_{2,n}$ in \cite{fang2025lightweight}\\
        $\textbf{T}_2 \gets \{T_{2,n}; n \in [1, N_m]\}$\\
       $ N_\text{opt} \gets \left| \left\{n \in [1, N_m]  \;:\; \textbf{T}_2[n] < 0 \right\} \right|$ \tcp{Count the number of negative values}
       \If{$\overline{\textbf{T}}_2\geq 0$\tcp{check the averged value}}{$\textbf{T}^\circ_2= \textbf{T}_2- \sqrt{\overline {\textbf{T}}_2}$}
       \Else{$\textbf{T}^\circ_2= \textbf{T}_2+ \sqrt{|\overline {\textbf{T}}_2|}$\tcp{compensate for $\eta$ miss-match}}
       $ N^\circ_\text{opt} \gets \left| \left\{n \in [1, N_m]  \;:\; \textbf{T}^\circ_2[n] < 0 \right\} \right|$\\
       $N_\text{opt} = \operatorname{round}(\frac{N_\text{opt}}{2} + \frac{N^\circ_\text{opt}}{2})$ \\
       $N_\text{opt} = \min(N_m,\max(N_\text{opt},3))$ \tcp{set upper and lower boundaries}
       }
    \end{algorithm}
    \section{Threat analysis of \gls{3gpp}\label{sec:prsthreatmodel} \gls{otdoa}-based localization}
  In the \gls{3gpp} localization framework, \gls{lpp} and \gls{nrppa} messages are secured through underlying network security mechanisms. \gls{lpp} is protected by \gls{nas}-level security, while \gls{nrppa} signaling is secured via \gls{tls} on the \gls{ngc} interface \cite{3gpp.38.401,3gpp.33.501}. However, \glspl{prs} lack encryption or authentication. These periodically broadcast downlink signals with limited bandwidth and power are inherently vulnerable to spoofing. A malicious actor can effectively spoof such narrowband signals. Since \glspl{prs} are received by multiple \glspl{ue} within a coverage area, spoofing attacks can significantly degrade localization accuracy or availability for all affected users, disrupting services at scale. We therefore model the \gls{prs} spoofing attack in the following subsection.
   \begin{algorithm}[!t]
    \caption{\gls{gd} algorithm}
    \label{alg:MAGD}
    \scriptsize
    \DontPrintSemicolon
     Input: Selected nodes set $\mathcal{U}$; initial \gls{gnss} position as $\hat{\bm{p}}$; learning rate $\alpha$ and discount factor $\beta$; momentum $m$; \gls{tdoa} distances $\hat{d}^n$; position of \gls{gnb} $\bm{p}^n$; maximum iteration $I$, convergence threshold $\theta_t$\\
    \For{$i = 1:I$}{    
        {    \For {$n = 1:N$}{ 
                $\tilde{\bm{d}}^n = \Vert \hat{\bm{p}} - \bm{p}^n\Vert$ }     
             $G^i \gets \sum\limits_{n\in\mathcal{U}}\frac{(\hat{\bm{p}} - \bm{p}^n) }{\hat{d}_n}\cdot(\hat{d}^n - \tilde{d}^n)$ \tcp*{Gradient}
             $D^i \gets \sum\limits_{n\in\mathcal{U}}(\hat{d}^n - \tilde{d}^n)$ \tcp*{Distance differences}
             update: $\hat{\bm{p}} \gets \hat{\bm{p}} + m\cdot \hat{\bm{p}} + \frac{{\alpha}}{N}\cdot \frac{G^i}{\Vert G^i \Vert} $; }\\
             \If{$D^i > D^{i-1}$}{$\alpha = \beta\cdot\alpha$}
             \If{$(D^i - D^{i-1})/D^i< \theta_t$}{break}}
   \end{algorithm}
   \subsection{\gls{tdoa} observation model}
   First, \gls{tdoa} measurements are based on correlating the \glspl{prs}, which are typically short pulses with good correlation properties. To ensure localization accuracy, \glspl{ue} are required to perform \gls{tdoa} measurement multiple times. We consolidate the true \gls{tdoa} values from all $N$ nodes and $K$ measurement rounds into a matrix $\bm{\tau}$:
   \begin{equation}\nonumber
   \bm{\tau} =
   \begin{bmatrix}
   \tau_1^1  & \tau_1^2 &\cdots &\tau_1^n & \cdots &\tau_1^N\\
   \tau_2^1 & \tau_2^2 &\cdots &  \tau_2^n& \cdots &\tau_2^N\\
  \vdots&\vdots&\ddots&\vdots&\ddots&\vdots\\
   \tau_k^1 & \tau_k^2 &\cdots &  \tau_k^n& \cdots &\tau_k^N\\
   \vdots&\vdots&\ddots&\vdots&\ddots&\vdots\\
   \tau_K^1 & \tau_K^2 &\cdots &  \tau_K^n& \cdots &\tau_K^N
   \end{bmatrix}.
   \end{equation}
   \gls{prs} transmissions are typically regulated with a fixed time interval $\Delta_t$. Given that clock drifts of \glspl{gnb} are generally very small according to \gls{3gpp} requirements, the node clock offsets $\bm{\delta} = [\delta_1, \delta_2, \ldots, \delta_N]$ remain approximately constant. The solved \gls{tdoa} matrix is given by:
   \begin{equation}
    \bm{\tau}_R = \bm{\tau} + \bm{\delta} \otimes \mathbf{1}^{K \times 1} + \bm{\epsilon}_D,
   \end{equation}
   where $\mathbf{1}^{K \times 1}$ is a $K$-dimensional column vector of ones. $\bm{\epsilon}_D \in \mathbb{R}^{K \times N}$ represent the random transmission delays due to wireless channel, in which $\epsilon_j^n\sim\mathcal{N}(0,\sigma^n)$. However, under potential spoofing attacks targeting the \glspl{prs}, some entries of $\bm{\tau}_R$ may be spoofed. The operator under spoofing can be defined as, 
   \begin{equation}
  \bm{\tau}_o(k,n) = \mathcal{S} [\bm{\tau}_R(k,n)],
  \end{equation}
   \begin{equation}
  \mathcal{S} [\bm{\tau}_R(k,n)] = 
  \begin{cases}
  \bm{\tau}_S(k,n), & \text{w.p. } P_S, \quad k \in \bm{S},\\
  \bm{\tau}_R(k,n), & \text{w.p. } 1-P_S, \quad k \in \bm{S},\\
  \bm{\tau}_R(k,n), & \text{otherwise}.
  \end{cases}
  \end{equation}
   where $\bm{S}$ represents the set of spoofed links, $P_S$ is the probability of a successful spoof attack for $k \in \bm{J}$, and $\bm{\tau}_S(k,n)$ indicates the spoofed \gls{tdoa}.

   \subsection{\gls{tdoa} spoofing model}\label{subsec: tdoasp}
   In \gls{tdoa} processing, the \gls{ue} reports the first significant correlation peak corresponding to the \gls{los} path; when multiple peaks exist, multiple \gls{tdoa} measurements can be provided. A spoofer using a single high-power pulse faces substantial risk: the genuine pulse remains detectable, and once the \gls{uav} identifies abnormal amplitude, extraction can preserve the authentic pulse. Since the real pulse's power and arrival time are unknown, guaranteeing success requires impractically high power. Consequently, single high-power pulses can be readily filtered. A sophisticated spoofer might instead employ multiple lower-power pulses that merge with the authentic pulse into a composite peak, masking genuine timing. However, since leading-edge detection is preferred in \gls{a2g} channels, spoofing pulses must arrive earlier and maintain sufficient power to dominate the merged peak's leading edge. Therefore, successful spoofing from a rogue \gls{uav} is: \begin{enumerate*}[label=\emph{\roman*)}]
    \item \textbf{power-limited}—insufficient power fails to make the spoofed leading edge detectable;
    \item \textbf{synchronization-limited}—imperfect synchronization, unknown propagation delay, and other factors create timing uncertainty.
    \end{enumerate*}

  Power limitations are well-studied and demonstrated in Subsection~\ref{subsec:resiloceva}. Here we focus on synchronization limitations. The correlation peak width is $l_m = 1/\beta$, where $\beta$ is signal bandwidth. While \gls{uav} synchronization is not required for \gls{otdoa}, the \gls{uav} synchronizes with the network for communication. Synchronization quality $\pm\sigma_u$ exhibits uncertainty from oscillator drift and network configuration. Re-synchronization is triggered when uncertainty exceeds a threshold \cite{3gpp29565}. We assume $\sigma_u$ is bounded by $\Delta_u$ and the spoofer's synchronization uncertainty is bounded by $\Delta_{sp}$. $P_s$ denotes spoofing success probability, while $\tau_u$ and $\tau_{sp}$ are propagation delays. To understand how these parameters affect attack success, we establish the following relationship:

  \begin{lemma}\label{lemma:2} Under typical conditions ($\Delta_u > l_s + l_m$, $\tau_{sp} < \tau_u$), $P_s$ increases with $l_m$ and \gls{ue} $\Delta_u$, while decreasing with $\Delta_{sp}$. These trends may reverse under exceptional geometric conditions.
  \end{lemma}
   \begin{proof}
   See Appendix~\ref{sec:proof3} for the detailed proof.
   \end{proof}
   \begin{figure}[!t]
    \centering
    \begin{subfigure}[b]{0.241\textwidth}
        \centering
        \includegraphics[width=\textwidth]{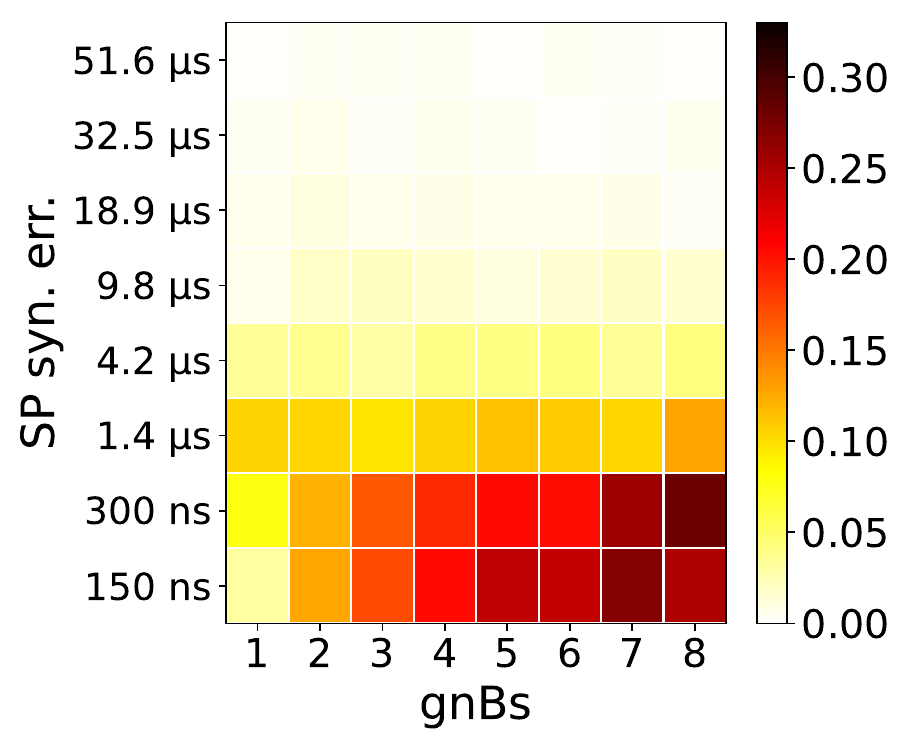}
        \caption{}\label{fig:bsPs}
    \end{subfigure}
   \begin{subfigure}[b]{0.241\textwidth}
        \centering
        \includegraphics[width=\textwidth]{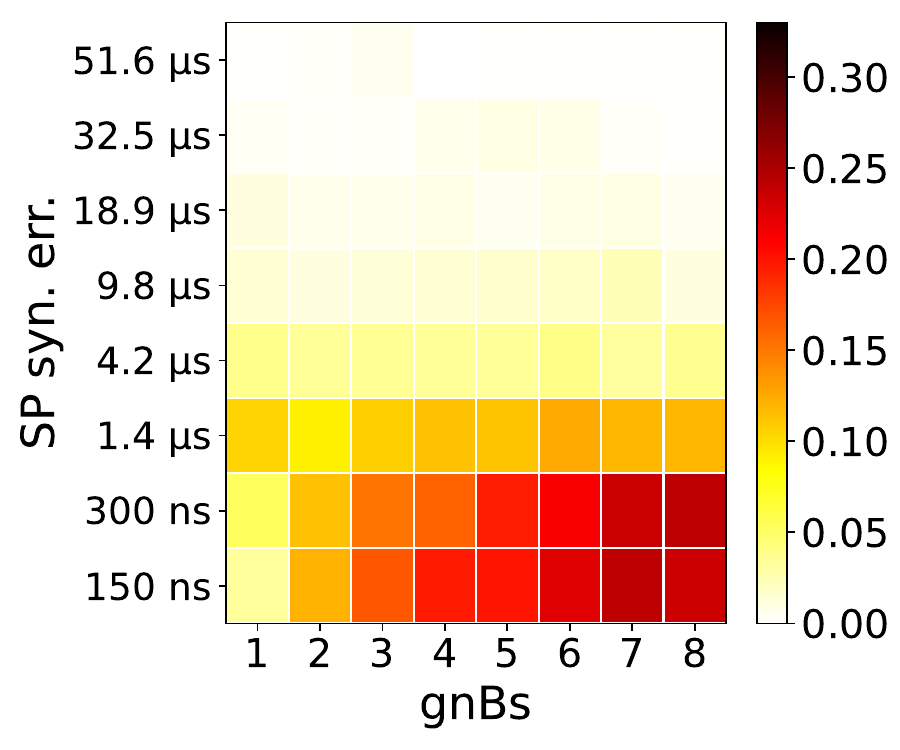}
        \caption{}\label{fig:DeltaPs}
    \end{subfigure}
   \begin{subfigure}[b]{0.241\textwidth}
        \centering
        \includegraphics[width=\textwidth]{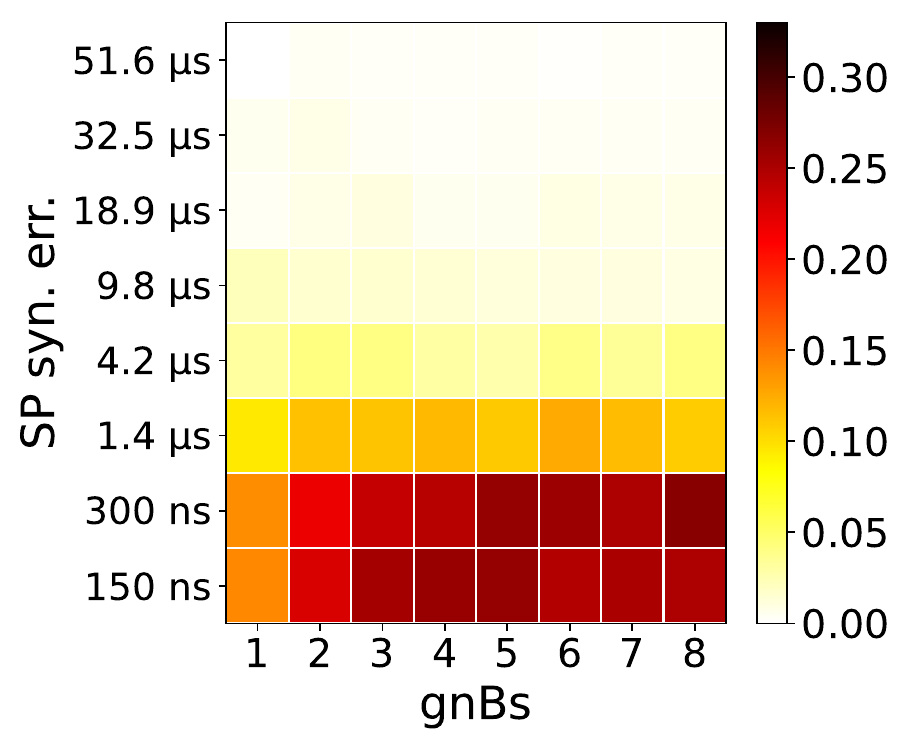}
        \caption{}\label{fig:tausPs}
    \end{subfigure}
   \begin{subfigure}[b]{0.241\textwidth}
        \centering
        \includegraphics[width=\textwidth]{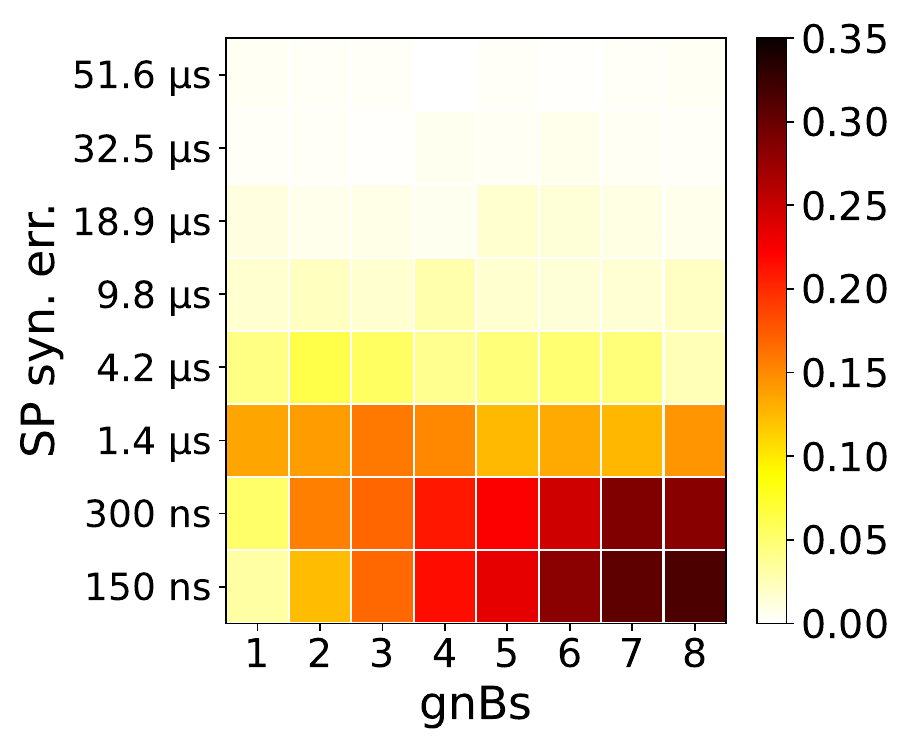}
        \caption{}\label{fig:lsPs}
    \end{subfigure}
   \caption{$P_s$ with varying parameters: (a) baseline scenario; (b) $\Delta_u = \SI{300}{\ns}$; (c) sending spoofing pulse earlier by $\SI{200}{\ns}$; (d) $l_s= \SI{200}{\ns}$.}
   \label{fig:Ps_var}
   \end{figure}   
   \subsubsection{Penetration Test} 
   Assuming the spoofer pulse is always strong enough to mask the authentic pulse, we evaluate $P_s$ through Monte Carlo simulation with $1{,}000$ iterations. The baseline parameters are $\Delta_u = \SI{1000}{\ns}$ and $l_m = \SI{50}{\ns}$. \gls{uav}-to-\gls{gnb} distances range from tens to hundreds of meters, while the spoofer is randomly positioned within $[10, 100]$ meters from the \gls{uav}. In Fig.~\ref{fig:Ps_var}, \glspl{gnb} are indexed 1 to 8 in ascending order of distance from the \gls{uav}. We set the minimum peak separation requirement to $l_s = 2l_m$, following the Rayleigh criterion for resolving closely-spaced correlation peaks \cite{kay1998fundamentals}. We vary $\Delta_u$ and $l_s$ relative to the baseline and evaluate early spoofing pulse transmission to enhance $P_s$ for nearby \glspl{gnb}.

   In general, $P_s$ benefits from tight spoofer synchronization, especially when \glspl{gnb} are distant. Comparing Figs.~\ref{fig:bsPs} and \ref{fig:DeltaPs}, improved \gls{uav} synchronization reduces $P_s$, particularly for distant \glspl{gnb}. For nearby \glspl{gnb}, this reduction is minimal, consistent with our analysis in Appendix~\ref{sec:proof3}. Comparing Figs.~\ref{fig:bsPs} and \ref{fig:tausPs}, early spoofing pulses significantly enhance $P_s$ for nearby \glspl{gnb} while reducing it for distant ones (e.g., \glspl{gnb} 7–8). In Fig.~\ref{fig:lsPs}, $P_s$ increases with larger $l_s$.

     \begin{figure}[!t]
   \begin{subfigure}[b]{0.241\textwidth}
        \centering
        \includegraphics[width=\textwidth]{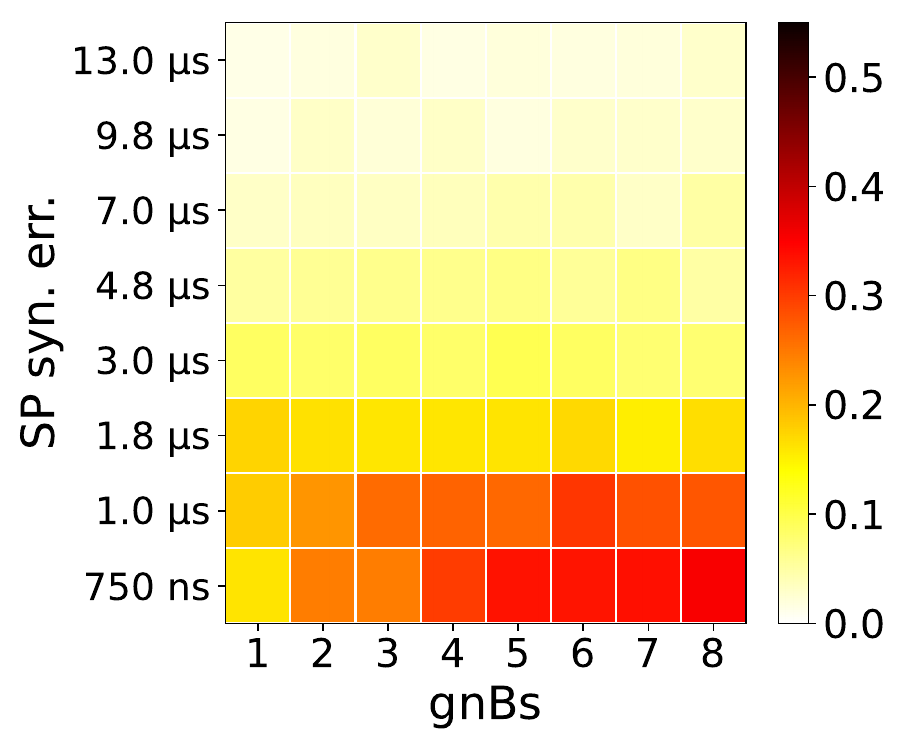}
        \caption{Pulse number $=3$}
    \end{subfigure}
   \begin{subfigure}[b]{0.241\textwidth}
        \centering
        \includegraphics[width=\textwidth]{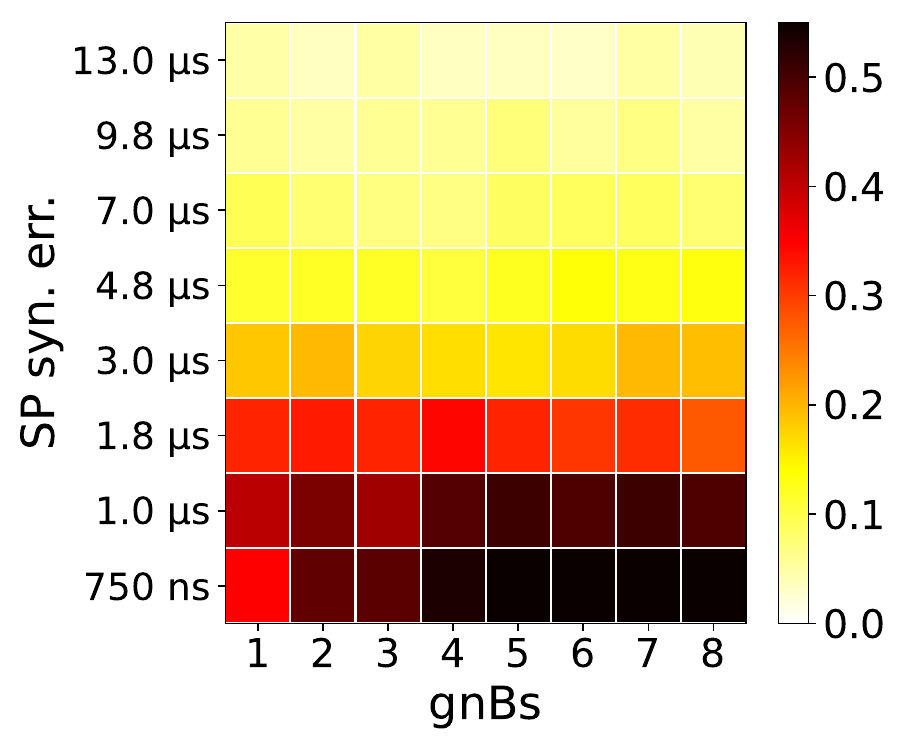}
        \caption{Pulse number $=5$}
    \end{subfigure}
   \caption{$P_s$ with increased spoofing pulses}
   \label{fig:spsvar}
   \end{figure}
   \subsubsection{Attack Strategies} 
   Achieving tight core network synchronization is challenging even for network participants (e.g., rogue \gls{uav}). An alternative multi-pulse approach is shown in Fig.~\ref{fig:spsvar}. However, spoofers can only estimate their own distances to \glspl{gnb}, not the \gls{uav}-to-\gls{gnb} distances needed for precise timing. When the \gls{uav} is close to the spoofer, the asummed distance estimates are more reliable; when distant, timing errors and higher power requirements reduce effectiveness. Therefore, spoofers primarily target nearby victim \glspl{uav}. Spoofer \gls{uav} power is constrained. While spoofing distant \glspl{gnb} achieves high $P_s$ with lower power, spoofing nearby \glspl{gnb} causes larger localization errors even under weighted localization algorithms. Accordingly, we study three attack strategies: \textit{focused attack} (targeting only the strongest link with full power), \textit{global attack} (targeting all links with water-filling power allocation), and \textit{selective attack} (targeting only weak links with water-filling power allocation).
      
   \subsection{LMF-Based resilient and spoofer localization}\label{subsec:resiĺoc}
   The key to resilient localization lies in distinguishing the spoofed \gls{tdoa} measurements from the actually received \gls{tdoa} measurements. The anomaly detection filter is denoted by:
  \begin{equation}
   \{\hat{\bm{\tau}}_S(k,n), \hat{\bm{\tau}}_R(k,n)\} = \mathcal{D} [\bm{\tau}_o(k,n)],
  \end{equation}
   where $\hat{\bm{\tau}}_S$ and $\hat{\bm{\tau}}_R$ are the filtered spoofed and actually received \gls{tdoa} measurements, respectively. The missing entries of $\hat{\bm{\tau}}_R$ can be filled by interpolation and used to localize the \gls{uav}. Meanwhile, $p_n$ and $\hat{\bm{\tau}}_S$ can be used to localize the spoofer.
   \begin{figure}[!t]
    \centering
    \includegraphics[width=0.99999999\linewidth]{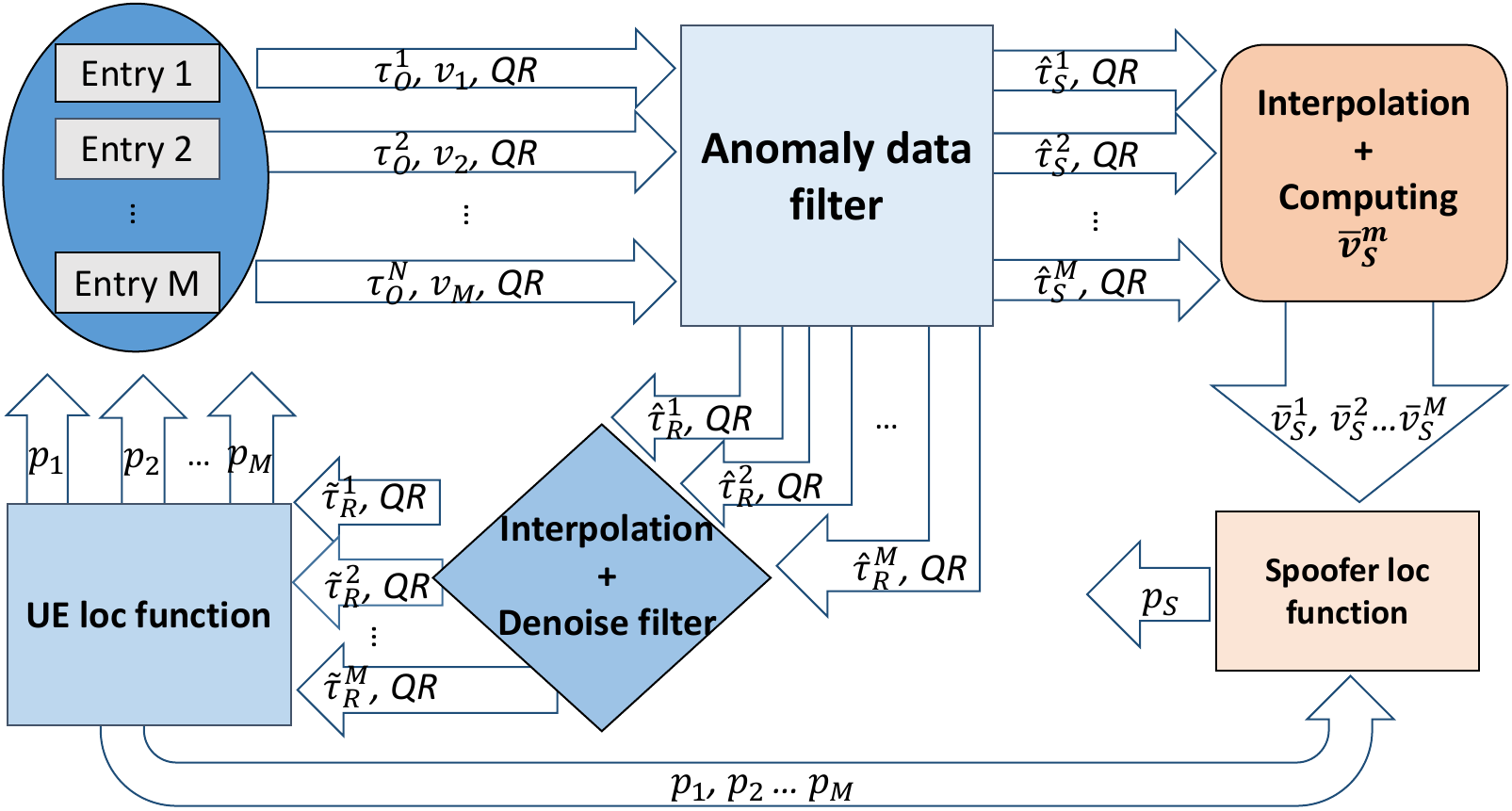}
    \caption{Diagram of resilient localization and spoofer localization)}\label{fig:resilocflow}
    \end{figure}
   A successful spoofing attack occurs when the spoofing pulses arrive ahead of the genuine pulse. Directly applying $\hat{\bm{\tau}}_S$ for distance estimation is infeasible due to unknown pulse transmission times. However, the \gls{uav} motion with respect to the spoofer generates nanosecond-level \gls{tdoa} variations. Since the spoofed pulse signature remains identical across transmissions under small variations, the relative motion can be estimated from the temporal differences:
  \begin{equation}
   v_S = \hat{\bm{\tau}}_S(k,n) - \hat{\bm{\tau}}_S(k-1,n).
  \end{equation}
  A single measurement of $v_s$ is highly unreliable; we consider the averaged result from $\hat{\bm{\tau}}_S$, denoted as $\overline{v_s}$, to be a valid measurement. In the localization system, multiple \glspl{uav} can be spoofed simultaneously, or a single \gls{uav} can use the localization service multiple times, with each instance considered as one data entry. We consider $M$ such data entries. The velocity and position of the $m$-th \gls{uav} are denoted by $\bm{v}^m$ and $\bm{p}^m$, respectively, while the spoofer position is $\bm{p}^s$. The radial velocity between the $m$-th \gls{uav} and the spoofer is given by:
  \begin{equation}
    v_S^m = \frac{\bm{v}^m \cdot (\bm{p}^m - \bm{p}^s)}{\|\bm{p}^m - \bm{p}^s\|}.
  \end{equation}
  The spoofer position $\bm{p}^s = [x^s, y^s, z^s]^\mathrm{T}$ is estimated by minimizing the residual between the measured and predicted radial velocities. For $M$ measurements, the optimization problem is formulated as:
\begin{equation}
    \bm{p}^s = \arg\min_{\bm{p}} \sum_{m=1}^{M} \left( v_S^m - \frac{\bm{v}^m \cdot (\bm{p}^m - \bm{p})}{\|\bm{p}^m - \bm{p}\|} \right)^2,
\end{equation}
where $v_s^m$ is the measured radial velocity for the $m$-th \gls{uav}, and $\bm{v}^m$ and $\bm{p}^m$ are the known velocity and position of the $m$-th \gls{uav}. The initial guess for the optimization is set to the centroid of all \gls{uav} positions. 

We consider the following anomaly data filter strategies: \subsubsection{\gls{tcv}} 
This approach validates triangular consistency, where the degree of consistency is determined by how many measurements fulfill the consistency check. This approach is simple, requires low computational overhead, and effectiveness in identifying inconsistencies in distance measurements, as demonstrated in \cite{securetmcwon2019}. For two random \glspl{gnb}, $i,j \in \{1,2,\ldots, N\}$, the consistency check is defined by:
\begin{equation}
  |\hat{d}_k^i - \hat{d}_k^j | - \epsilon_t \leq d^{i,j} \leq \hat{d}_k^i + \hat{d}_k^j + \epsilon_t,
\end{equation}
where $\hat{d}_k^i$ and $\hat{d}_k^j$ are the measured distances to different \glspl{gnb}, $d^{i,j}$ is the distance between the two \glspl{gnb}, and $\epsilon_t = 1.97(\sigma_{\mathrm{SNR}}^i+\sigma_{\mathrm{SNR}}^j)$ is the tolerance threshold based on the \gls{snr}-dependent distance measurement standard deviation. Measurements failing a single consistency check are filtered as anomalies.
\subsubsection{\gls{sdet}} 
In this approach, we use the imprecise \gls{gnss} position to validate the current distance measurement. The baseline distance is given by $d_\mathrm{bs} = \|\bm{p}_\mathrm{GNSS} - \bm{p}^n\|$, and the tolerance threshold is $\epsilon_t = 1.97(\sigma_{\mathrm{SNR}}^n+\sigma_{\mathrm{GNSS}})$. The static error check is determined by:
\begin{equation}
 d_\mathrm{bs} - \epsilon_t \leq \hat{d}^n \leq d_\mathrm{bs} + \epsilon_t.
\end{equation}
\subsubsection{\gls{rdef}} This approach batches the data by time step $k$. Initially, the first batch is filtered by \gls{sdet}, then the \gls{gd} block estimates the position and returns it to \gls{sdet} for verification, as depicted in Fig.~\ref{fig:RDEF}. The current estimate combined with the reported \gls{uav} velocity provides a warm-start initialization for the next time step, ensuring small variations from the true position and enabling convergence with fewer iterations and a smaller learning rate. As localization accuracy improves iteratively, the tolerance threshold is adaptively tightened: $\epsilon_t = 1.97(\sigma_{\mathrm{SNR}}^n+\beta_t\sigma_{\mathrm{GNSS}})$, where $\beta_t$ is the reduction factor.
\begin{figure}[!htbp]
    \centering
    \includegraphics[width=0.99999999\linewidth]{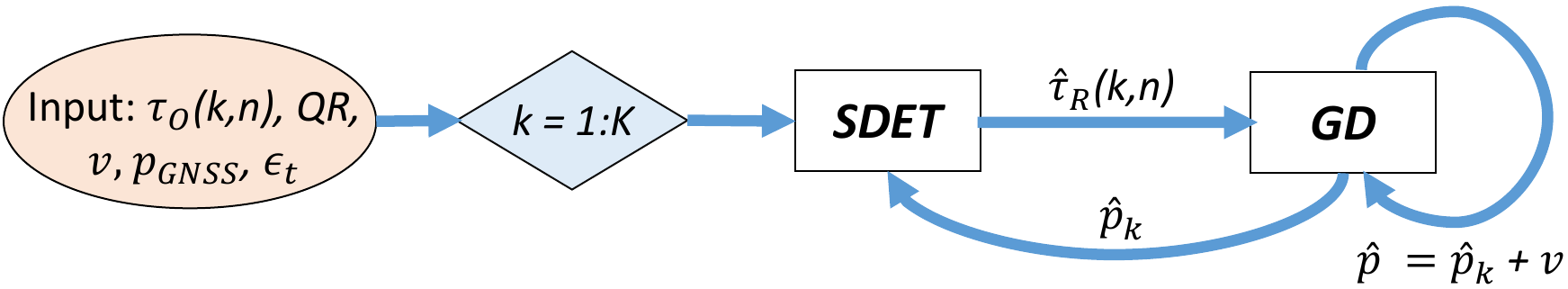}
    \caption{Diagram of resilient localization and spoofer localization)}\label{fig:RDEF}
    \end{figure}


   \section{Simulation results}\label{sec:simu}
   \subsection{Evaluation of node selection strategies}
     \begin{figure*}[!htbp]
    \centering
    \begin{subfigure}[b]{0.30\textwidth}
        \centering
        \includegraphics[width=\linewidth]{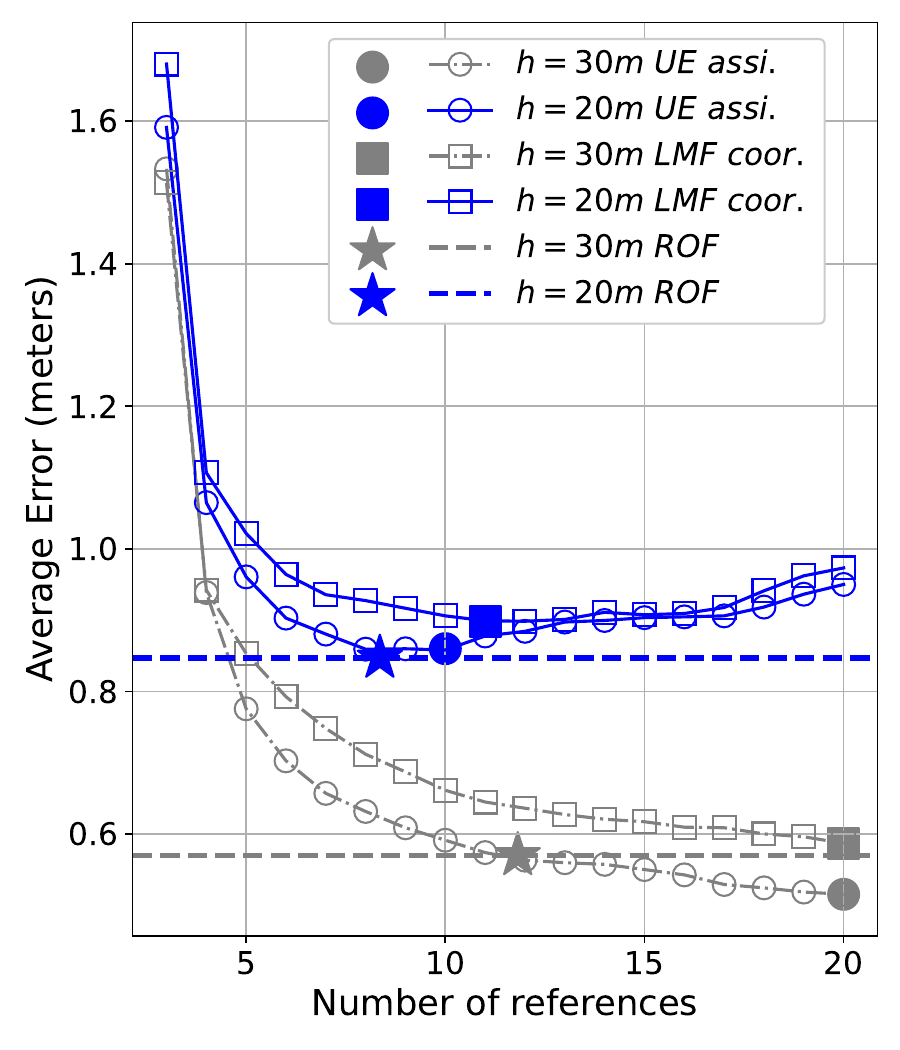}
        \caption{ $R = \SI{60}{\meter}$}
        \label{fig:cl}
    \end{subfigure}
    \begin{subfigure}[b]{0.30\textwidth}
        \centering
        \includegraphics[width=\linewidth]{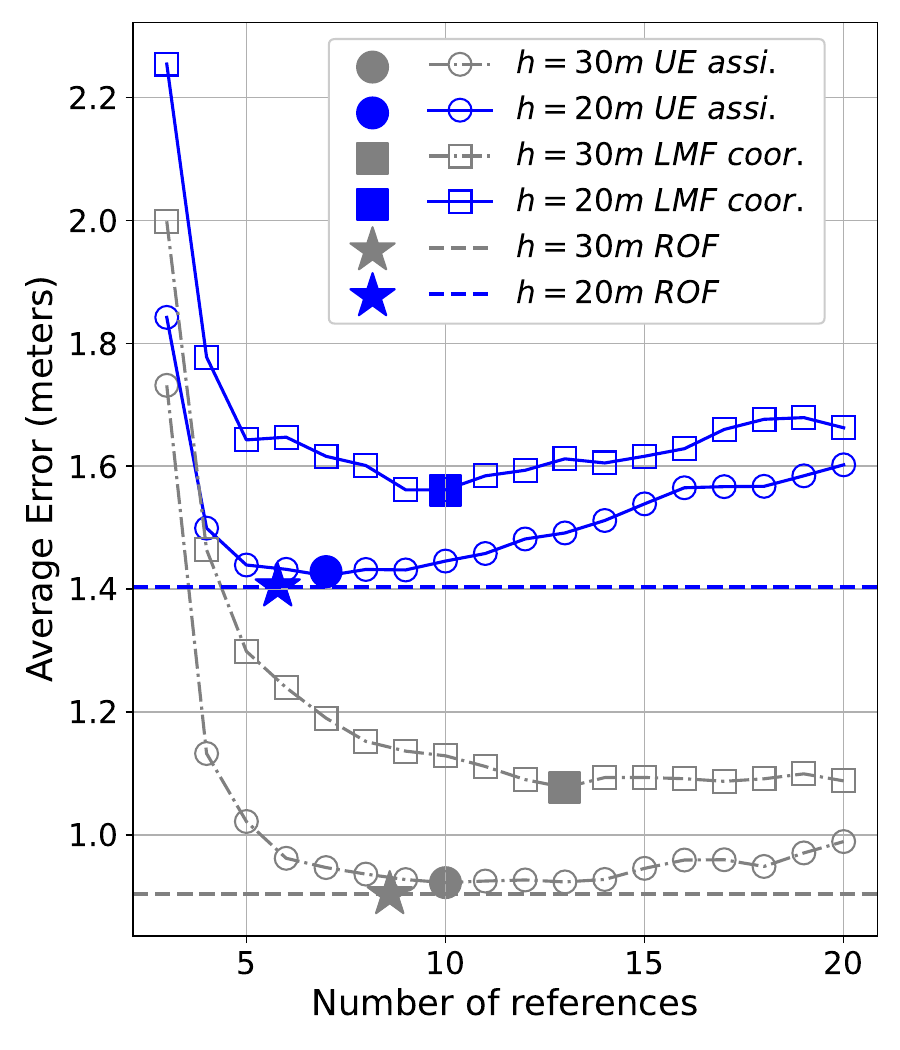}
        \caption{$R = \SI{90}{\meter}$}
        \label{fig:randomdep}
    \end{subfigure}
    \begin{subfigure}[b]{0.30\textwidth}
        \centering
        \includegraphics[width=\linewidth]{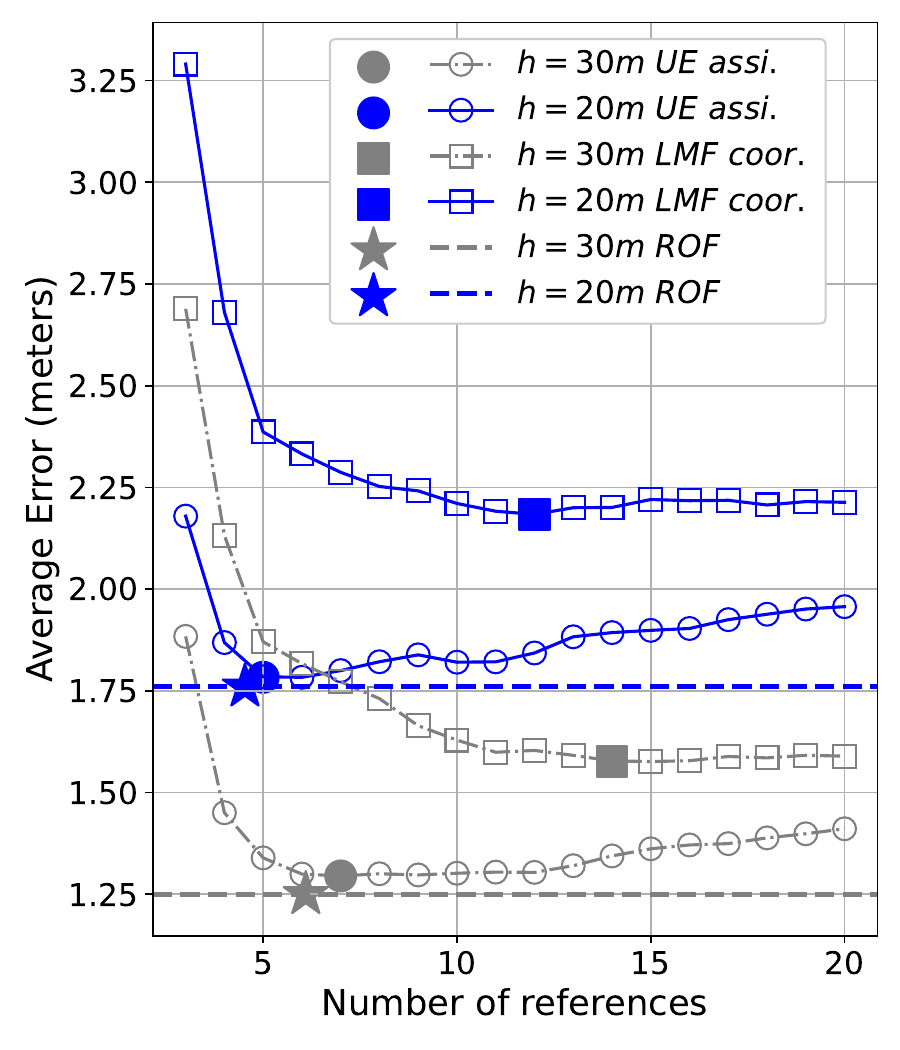}
        \caption{$R = \SI{120}{\meter}$}
        \label{fig:coordinateddep}
    \end{subfigure}    
    \caption{\gls{ue} Localization error of node selection under varying $h$ and $R$ (optimal $N$ are marked with solid markers).}
    \label{fig:optsimresu}
\end{figure*}

  First, we compare the performance of \gls{rof} with the aforementioned approaches under varying \gls{gnb} densities and \gls{uav} altitudes. We assume that unilateral reference issues are resolved by the \gls{lmf}. All nodes are synchronized, and any residual synchronization offsets are calibrated by the \gls{lmf} during the localization process. Table~\ref{tab:setup1} lists the detailed simulation setup, including channel parameters, signal properties, deployment configurations, and gradient descent parameters, to ensure reproducibility. The delay spread, Rician factors, and number of multipaths were adapted from \cite{3gpp.38.901} to ariel scenarios, while the transmit power corresponds to a small-cell setup and the noise floor accounts for both thermal and environmental interference noise.
      \begin{table}[!t]
		\centering
        \scriptsize
		\caption{Simulation setup 1}
		\label{tab:setup1}
		\begin{tabular}{>{}m{0.2cm} | m{1.6cm} l m{3.7cm}}
			\toprule[2px]
			&\textbf{Parameter}&\textbf{Value}&\textbf{Remark}\\
			\midrule[1px]        
			&$f_c$&$3.5$ Ghz& Carrier frequency\\

			&$K$&$(0.1,3.0)$& Rician factors\\
			&$N_p$&$4$& Number of multipath\\
            & $\tau_\text{max}$ & 2e-7 s &  Maximum delay spread\\ 
            & $P_\text{t}$ & 15 dBm &  Transmitting power\\ 
            & $N_\text{o}$ & -91 dBm &  Noise floor\\ 
            & $\beta_n$ & 10 Mhz &  Bandwidth\\ 
            & $\sigma_t$ & $1 \mu s $&  Average synchronization error\\
            \midrule[1px]
            \multirow{-9.9}{*}{\rotatebox{90}{\textbf{TDOA}}}
            & $h_u $&$[20,30]$& \gls{uav} altitude\\ 
            &$\sigma_h$&$1$& Altitude standard deviation\\
			&$\sigma_\text{GPS}$&$5$& Initial GPS standard deviation\\
            &$R$&$[60,90,120]$& Node coverage\\
            &$h_n$&$\sim\mathcal{U}(0,5)$& Node Altitudes\\
			\multirow{-5.2}{*}{\rotatebox{90}{\textbf{Deployment}}}
            & $\Delta_\text{LOS}$ & $[-0.4,0.1]$&  \gls{los} probability modification\\
            & $d_{gc} $&$\sim\mathcal{U}(0,60)$& \gls{uav} distances to the geometry center\\ 
            \midrule[1px]
            & $\alpha; \beta$ & $4;0.5$&  Learning rate and discount factor\\
            & $m$ & 1e-5&  Momentum\\
            & $I$ & $50$&  Maximum iteration\\
            \multirow{-4.2}{*}{\rotatebox{90}{\textbf{GD}}} & $\theta_t$ & 4e-5&  Convergence threshold\\
            \bottomrule[2px]
		\end{tabular}
	\end{table}
     However, involving a large number of references is often impractical. Therefore, we limit the number of reference nodes to $20$. The aggregated results from $1,000$ simulations are shown in Fig.~\ref{fig:optsimresu}.
  
   The simulation results confirm our analytical findings: localization considering node distribution has an optimal solution regardless of whether \gls{lmf} coordinated or \gls{ue} assisted node selection is employed. $N_\text{opt}$ varies with respect to node coverage and altitude. While $N$ is limited to $20$, \gls{ue} assisted selection generally outperforms \gls{lmf} coordinated selection, but exhibits a steep gradient when increasing beyond its minimum point.
   Despite higher altitude $h_u = \SI{30}{\meter}$ resulting in generally longer distances to \glspl{gnb}, localization performance remains superior due to improved channel conditions. Algorithm~\ref{alg:rbof} demonstrates slight improvements in localization accuracy while requiring fewer reference nodes by providing dynamic solutions for different scenarios. Except for the case where $h_u = \SI{30}{\meter}$ and $R = \SI{60}{\meter}$, the average optimal solution consistently exceeds $20$ nodes. In this configuration, Algorithm~\ref{alg:rbof} underperforms compared to $N=20$. While the true $N_\text{opt}$ often surpasses $20$, the imposed upper bound of $20$ limits achievable performance. Nevertheless, substantial increases in $N$ provide diminishing returns in performance enhancement. Importantly, the existing localization performance already meets satisfactory standards for most cases.
   \begin{figure}[!t]
    \centering
    \includegraphics[width=0.84\linewidth]{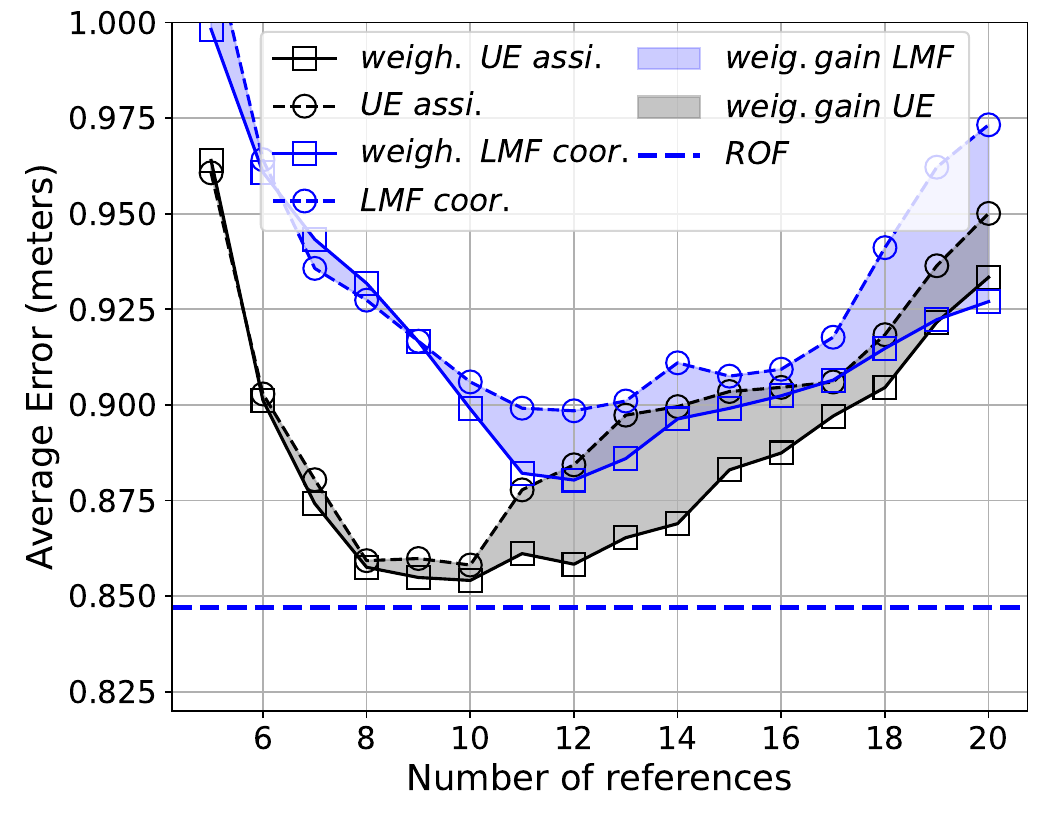}
    \caption{\gls{ue} localization error with enhancement of weighted approach ($R=\SI{60}{\meter}$, $h=\SI{30}{\meter}$)}\label{fig:weighted}
    \end{figure}
    
   Second, we investigate performance enhancement when a weighted approach is applied. When the \gls{ue} returns measurement reports, \gls{rssi} values can also be reported to the \gls{lmf} to enable weighted localization for performance improvement. The weight set is calculated based on the \gls{rssi} set $\boldsymbol{\gamma}_R = \{\text{SNR}_1,\ldots,\text{SNR}_n,\ldots,\text{SNR}_N\}$, given by $\mathbf{w} = N {\boldsymbol{\gamma}_R}/ \sum {\boldsymbol{\gamma}_R}$. The performance comparison between weighted and unweighted approaches is depicted in Fig.~\ref{fig:weighted}. The weighted approach enhances both node selection strategies, but performance degrades when including more distant nodes. For \gls{lmf} coordinated localization, performance is enhanced across the entire domain, though the enhanced optimal performance is still outperformed by \gls{ue} assisted localization. For \gls{ue} assisted localization, the weighted approach provides substantial performance gains when $N$ is not optimized, but offers only marginal improvements when $N$ is already optimized. Overall, \gls{rof} continues to outperform the weighted approach.  
   \begin{figure}[!t]
    \centering
   \begin{subfigure}[b]{0.2115\textwidth}
        \centering
        \includegraphics[width=\textwidth]{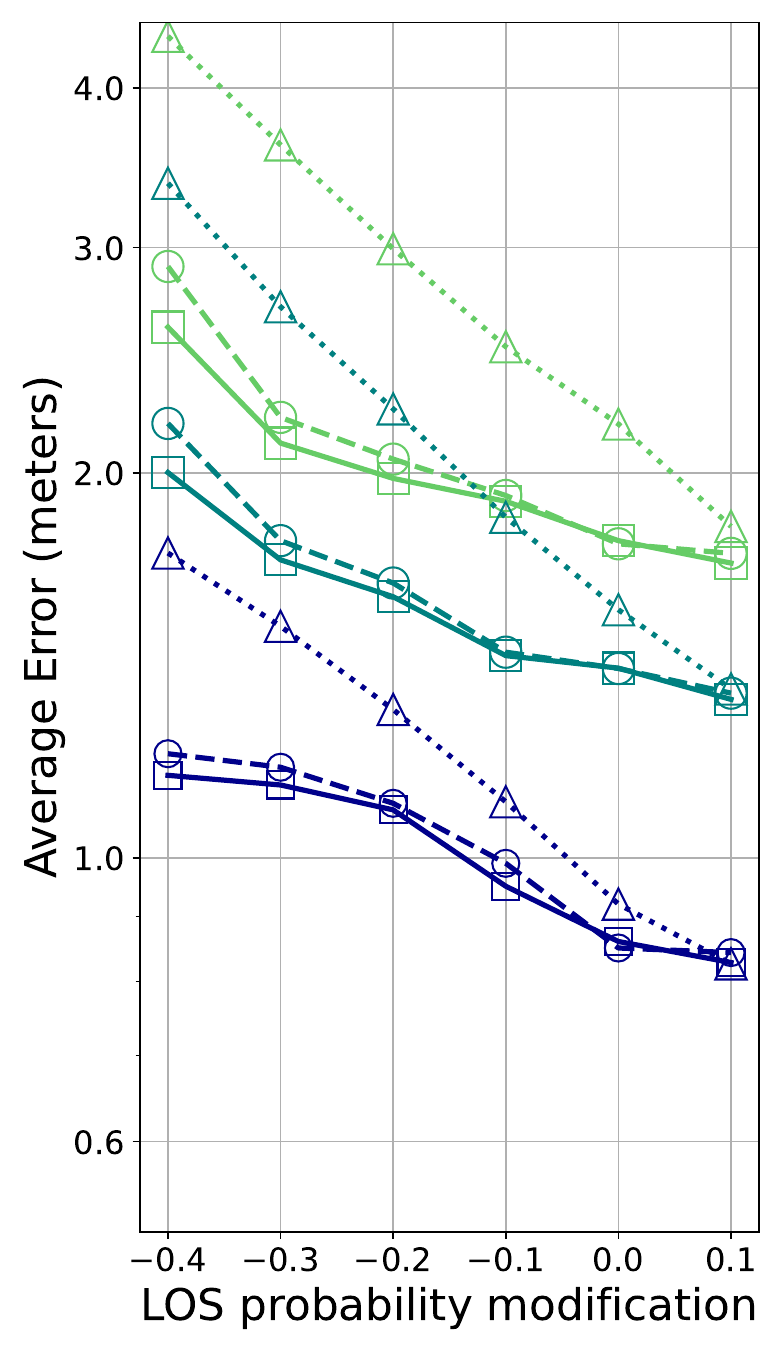}
        \caption{$h = \SI{20}{\meter}$}
    \end{subfigure}
   \begin{subfigure}[b]{0.269\textwidth}
        \centering
        \includegraphics[width=\textwidth]{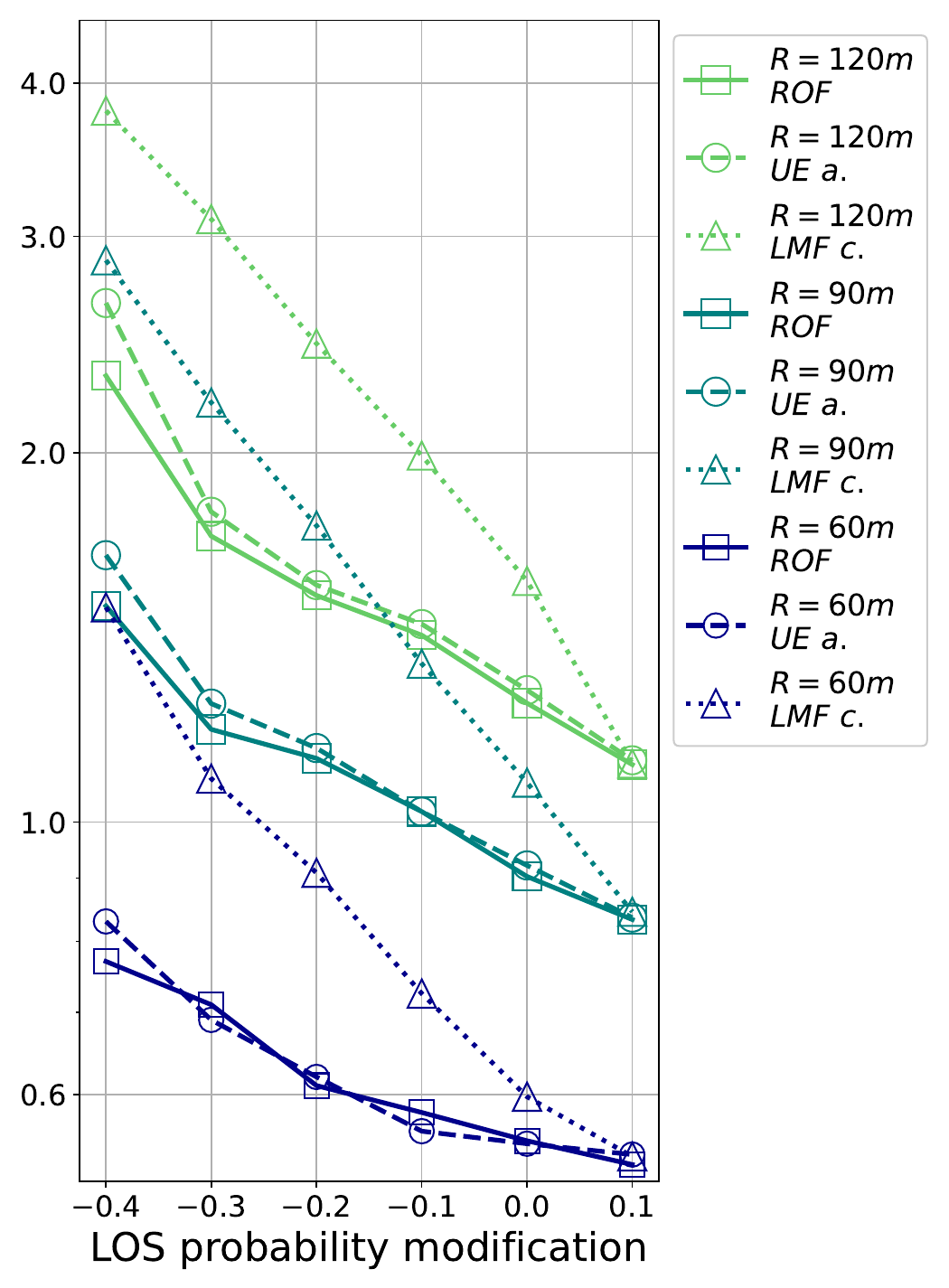}
        \caption{$h = \SI{30}{\meter}$}
    \end{subfigure}
   \caption{\gls{ue} Localization error of \gls{los} variation}
   \label{fig:losvar}
   \end{figure}
   Third, assuming the \gls{uav} has precise knowledge of its altitude and ground reference node positions, a simple empirical optimal solution could theoretically be employed. In practice, however, environmental factors severely affect \gls{los} probability, making empirical approaches unreliable even when altitude errors and deployment uncertainties are ignored. To evaluate this limitation, we modify $P_\text{los}$ according to Eq.~(\ref{eq:problos}) and benchmark the empirical solution against our \gls{rof}-based optimization. Fig.~\ref{fig:losvar} reveals distinct behavioral patterns: \gls{lmf} coordinated localization shows linear degradation with decreasing \gls{los} probability, whereas \gls{ue} assisted localization maintains robustness under poor \gls{los} conditions. Notably, \gls{rof} demonstrates superior performance compared to empirical methods across all scenarios, with particularly pronounced advantages in challenging \gls{los} environments.

   Third, assuming the \gls{uav} has precise knowledge of its altitude and ground reference node positions, a simple empirical optimal solution could theoretically be employed. In practice, however, environmental factors severely affect \gls{los} probability, making empirical approaches unreliable even when altitude errors and deployment uncertainties are ignored. To evaluate this limitation, we modify $P_\text{los}$ according to Eq.~(\ref{eq:problos}) and benchmark the empirical solution against our \gls{rof}-based optimization. Fig.~\ref{fig:losvar} reveals distinct behavioral patterns: \gls{lmf} coordinated localization shows linear degradation with decreasing \gls{los} probability, whereas \gls{ue} assisted localization maintains robustness under poor \gls{los} conditions. Notably, \gls{rof} demonstrates superior performance compared to empirical methods across all scenarios, with particularly pronounced advantages in challenging \gls{los} environments.

  \subsection{Evaluation of resilient localization}\label{subsec:resiloceva}
  We consider a scenario with one \gls{uav} and one spoofer to isolate and demonstrate the fundamental attack-defense dynamics. Extension to multiple \glspl{uav} involves consensus optimization and is left for future work. The evaluation proceeds in two phases: the configuration is optimized under benign conditions; then, the spoofer initiates attacks on the localization process. The \gls{lmf} simultaneously performs secure victim localization while recording anomaly data for spoofer localization. Both the spoofer and \gls{uav} are randomly positioned near the center of the deployment area. Simulation parameters are listed in Tab.~\ref{tab:setup2} for reproducibility; parameters unchanged from Tab.~\ref{tab:setup1} are omitted. The simulation results are based on $1,000$ Monte Carlo runs. Under the described deployment setup, the average optimized node selection count is
  $8$. Compared to previous evaluations, the localization performance is significantly improved in the absence of attacks, thanks to multiple rounds of \gls{tdoa} measurements.
  \begin{table}[!t]
		\centering
        \scriptsize
		\caption{Simulation setup 2}
		\label{tab:setup2}
		\begin{tabular}{>{}m{0.2cm} | m{1.6cm} l m{3.7cm}}
        \toprule[2px]
            & $h_u $&$25$& \gls{uav} altitude\\ 
			&$\sigma_\text{GPS}$&$\sim\mathcal{U}(3,7)$& Initial GPS error standard deviation\\
            &$R$&$100$& Node coverage\\
            &$\Delta_t$&$50$ ms& Measurement interval\\
            \multirow{-5.2}{*}{\rotatebox{90}{\textbf{Config}}}&$K$&$10$& Measurement rounds\\
            \midrule[1px]
			& $d_{gc} $&$\sim\mathcal{U}(0,60)$& Spoofer and legitimate \gls{uav} distances to the geometry center\\ 
            & $\theta_{d} $&$-10$ dB& \gls{tdoa} detection threshold with respect to maximum peak\\ 
            & $\Delta_{u}, \Delta_{sp} $&$1000$ ns & Synchronization error upper bound\\ 
             \multirow{-5.2}{*}{\rotatebox{90}{\textbf{Spoofer}}}& $N_p $&$5$  & Spoofing pulse number\\ 
             \midrule[1px]
            & $M $&$20$  & Data entries\\ 
            & $\alpha $&$1.5$  & Learning rate\\
            & $I_m $&$5$  & Maximum iteration\\
            \multirow{-4.2}{*}{\rotatebox{90}{\textbf{RDEF}}}& $\beta_t $&$0.97, 0.99$  & Reduction factor \gls{rdef}\\
            \bottomrule[2px]
		\end{tabular}
	\end{table}
 \begin{figure}[!t]
    \centering
    \begin{subfigure}[b]{0.241\textwidth}
        \centering
        \includegraphics[width=\textwidth]{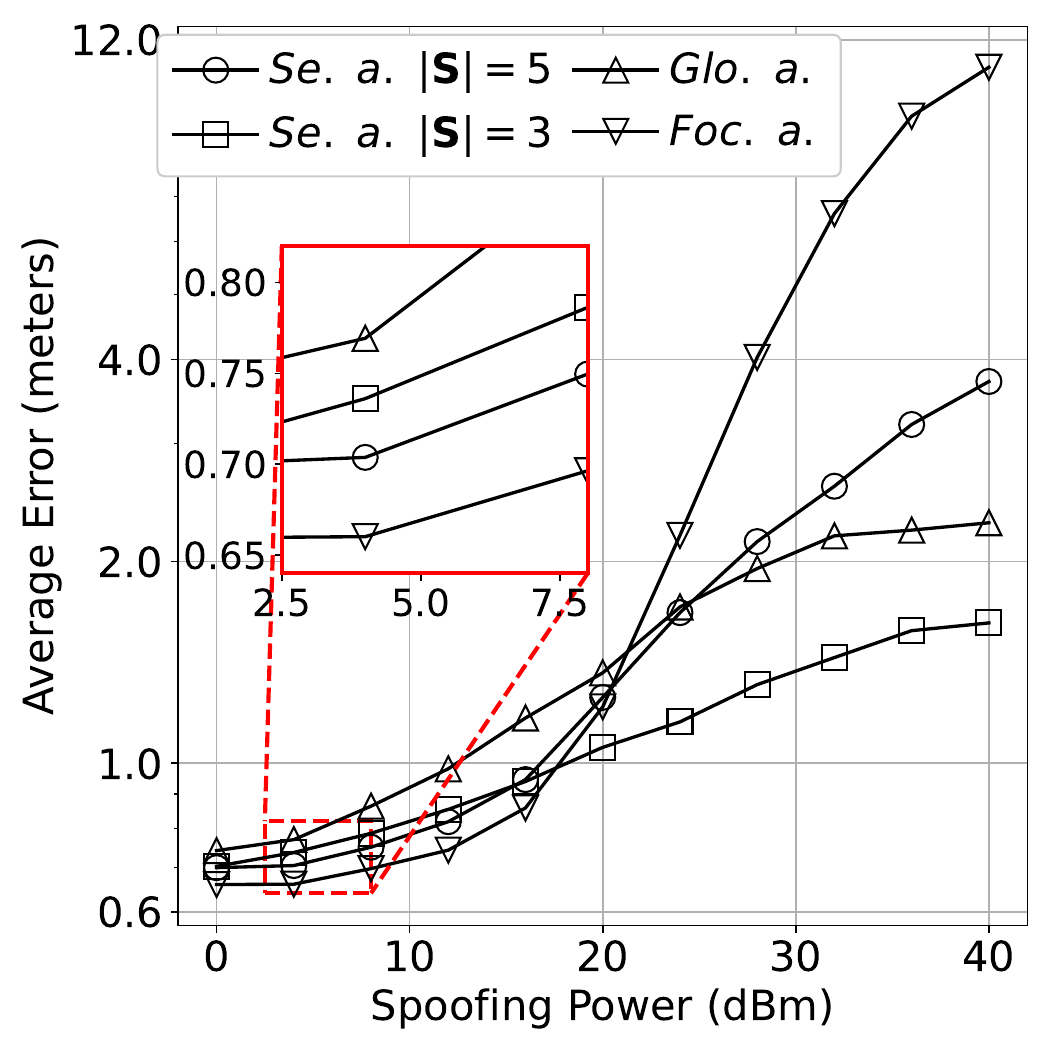}
        \caption{Baselines}\label{fig:bsat}
    \end{subfigure}
   \begin{subfigure}[b]{0.241\textwidth}
        \centering
        \includegraphics[width=\textwidth]{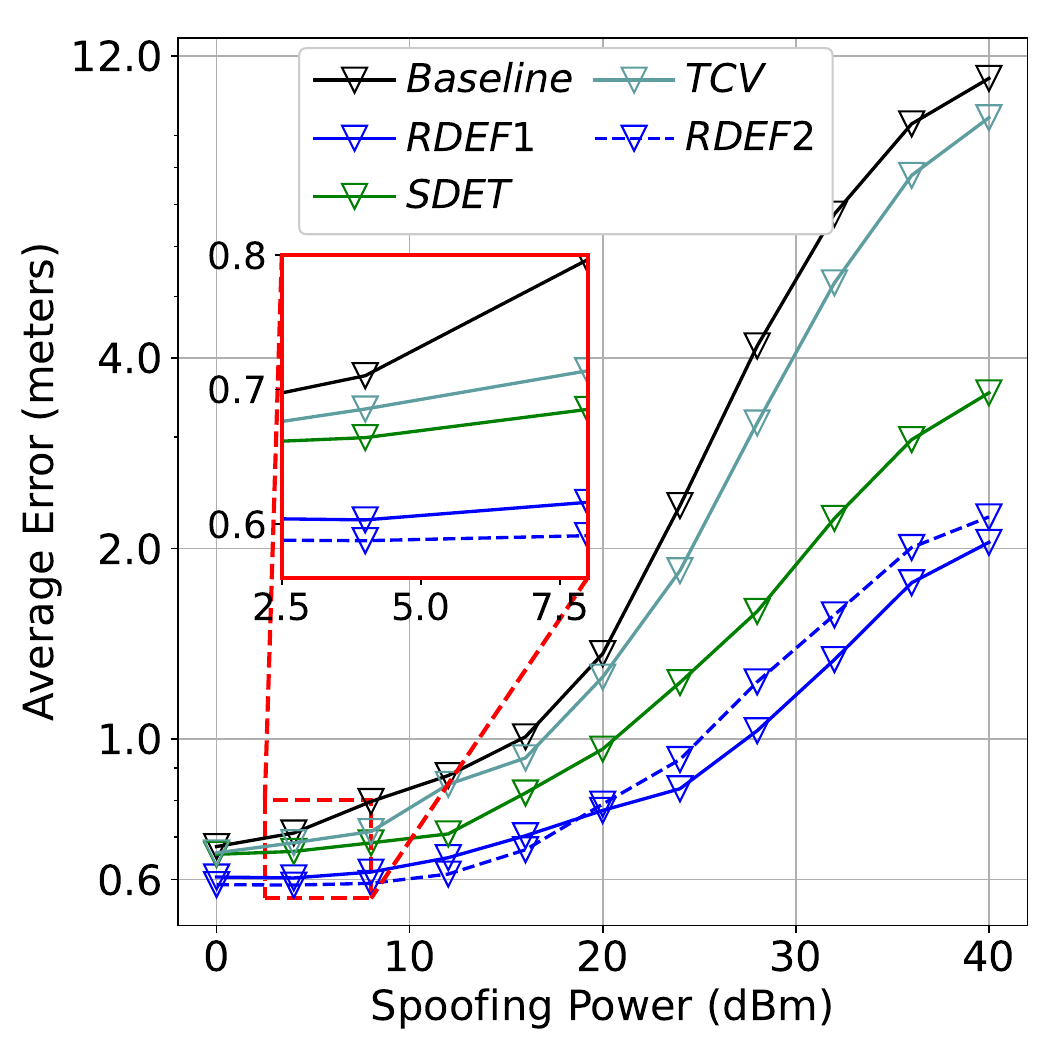}
        \caption{\emph{Global attack}}\label{fig:bsglo}
    \end{subfigure}
   \begin{subfigure}[b]{0.241\textwidth}
        \centering
        \includegraphics[width=\textwidth]{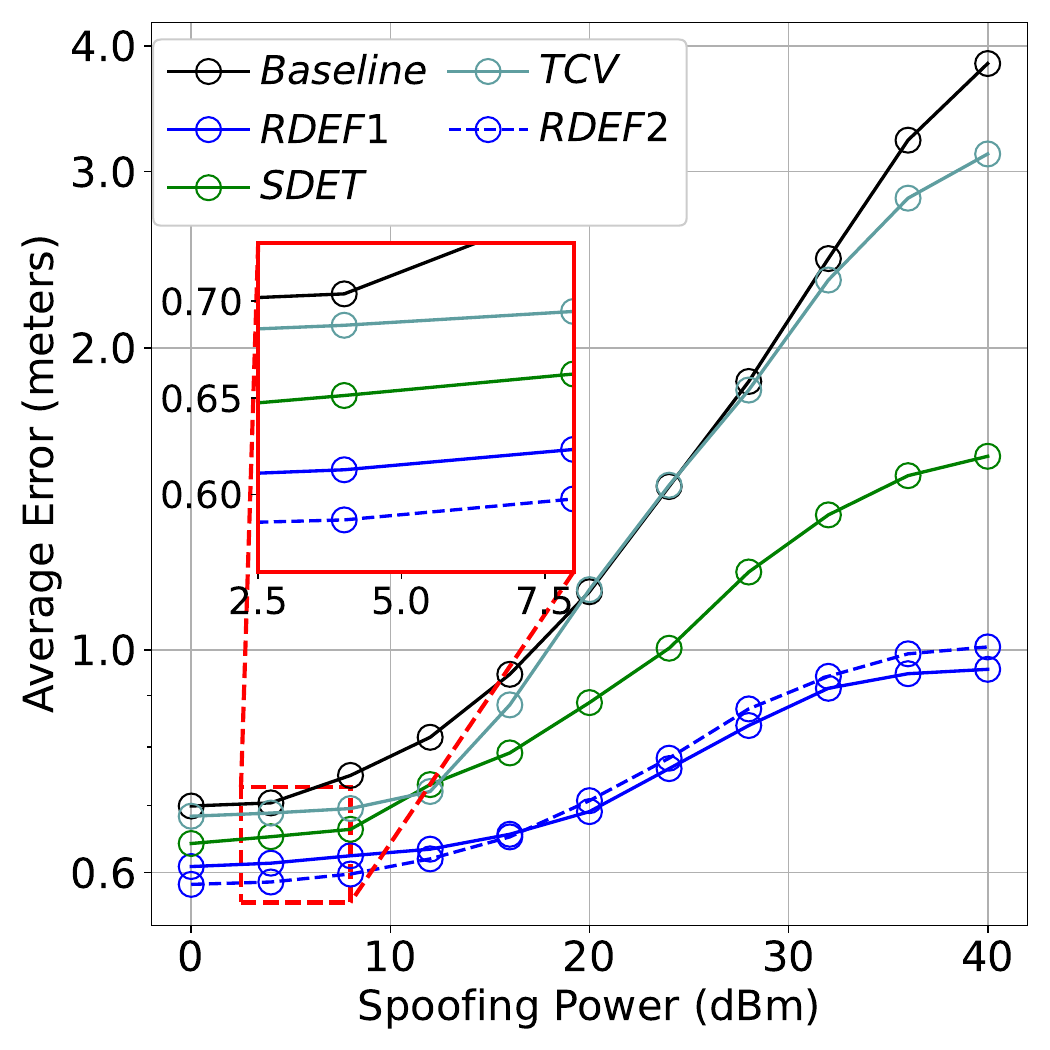}
        \caption{\emph{Selective attack} $|\bm{S}| = 5$}\label{fig:bssel}
    \end{subfigure}
   \begin{subfigure}[b]{0.241\textwidth}
        \centering
        \includegraphics[width=\textwidth]{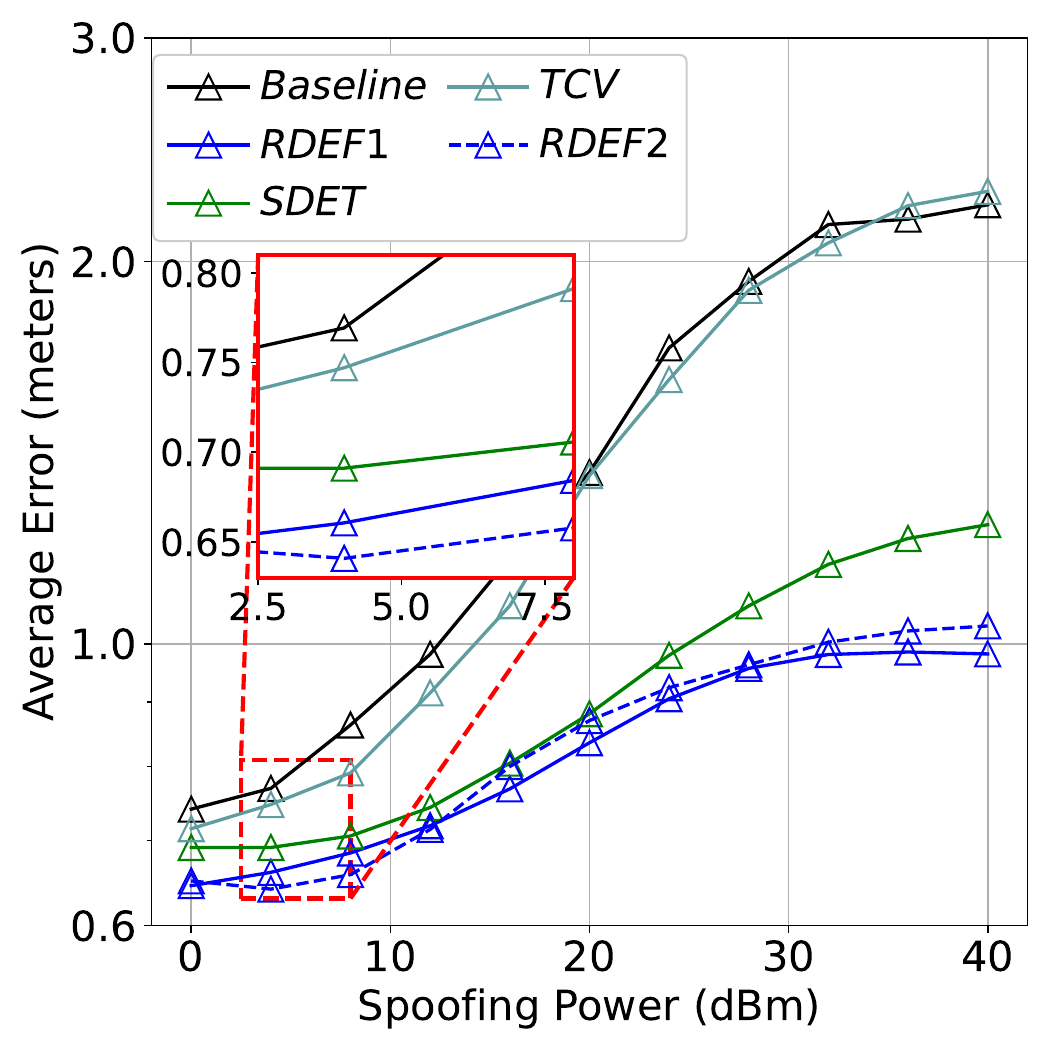}  
        \caption{\emph{Focused attack}}\label{fig:bsfoc}
    \end{subfigure}\caption{\gls{ue} localization error under attacks with or without anomaly filtering. \emph{TCV} is from \cite{securetmcwon2019}, while \emph{SDET} is from \cite{GD2025fang}, \emph{RDEF} is the proposed approach.}\label{fig:ls_Ps}, 
    \end{figure}
  
  First, we evaluate the attack effect without any anomaly filtering strategies, where the spoofer and \glspl{uav} are randomly positioned around the geometry center. For \emph{selective attack}, we evaluate two cases with spoofed link sets $|\bm{S}| = 3, 5$. In Fig.~\ref{fig:bsat}, in the low power regime (spoofing power $< 10$ dBm), the attack effectiveness follows: $\emph{Focused attack} > \emph{Selective attack} \:|\bm{S}| = 3 > \emph{Selective attack}\:|\bm{S}| = 5 > \emph{Global attack}$. This indicates that with limited power, a widely spread attack causes spoofing pulses to fail surpassing the \gls{tdoa} detection threshold, making the attack \textbf{power-limited}. In the high power regime (spoofing power $> 30$ dBm), the effectiveness becomes: $\emph{Global attack} > \emph{Selective attack} \:|\bm{S}| = 5> \emph{Focused attack} > \emph{Selective attack} \:|\bm{S}| = 3 $. As analyzed earlier, the weighted approach gives stronger links more influence in localization. Moreover, in Alg.~\ref{alg:MAGD}, the gradient is scaled by distance, meaning malicious gradients from smaller $\tilde{d}_n$ induce larger localization errors. However, $\emph{Global attack}$ and $\emph{Selective attack} \:|\bm{S}| = 5$ outperform $\emph{Focused attack}$ due to significantly more spoofed \gls{tdoa} measurements. Additionally, $\emph{Focused attack}$ converges in this regime, indicating it becomes \textbf{synchronization-limited}, where imperfect synchronization and geometry mismatch constrain further gains despite increased power. In Figs.~\ref{fig:bsglo}-\ref{fig:bsfoc}, we validate the performance of anomaly data filtering strategies. For \gls{rdef} approach, we define $\emph{RDEF1}$ for $\beta_t = 0.97$, and $\emph{RDEF2}$ for $\beta_t = 0.99$. 
  In the high power regime, we have suppression effect $\emph{RDEF1} > \emph{RDEF2} > \emph{SDET} > \emph{TCV}$ and in low power regime, we have $\emph{RDEF2} > \emph{RDEF1} > \emph{SDET} > \emph{TCV}$ . For \emph{TCV}, its supression effect is minimal, because the injected distance error is generally small that it bypass the consistency check. For \emph{RDEF}, it outperforms \gls{sdet} by continuously calibrating the verification process. 

     \begin{figure}[!t]
    \centering
   \begin{subfigure}{0.241\textwidth}
        \centering
        \includegraphics[width=\textwidth]{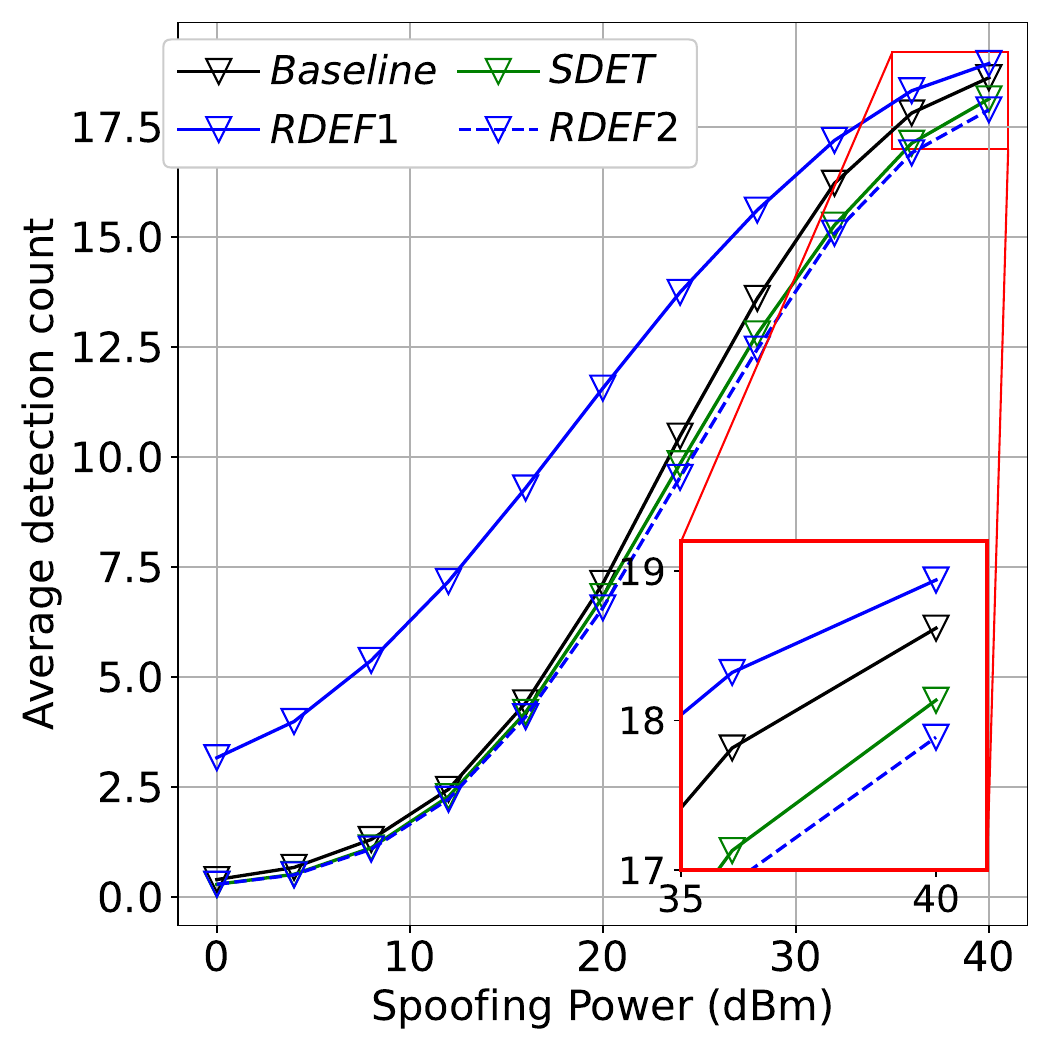}
        \caption{\emph{Global attack}} \label{fig:Gcount}
    \end{subfigure}
   \begin{subfigure}[b]{0.241\textwidth}
        \centering
        \includegraphics[width=\textwidth]{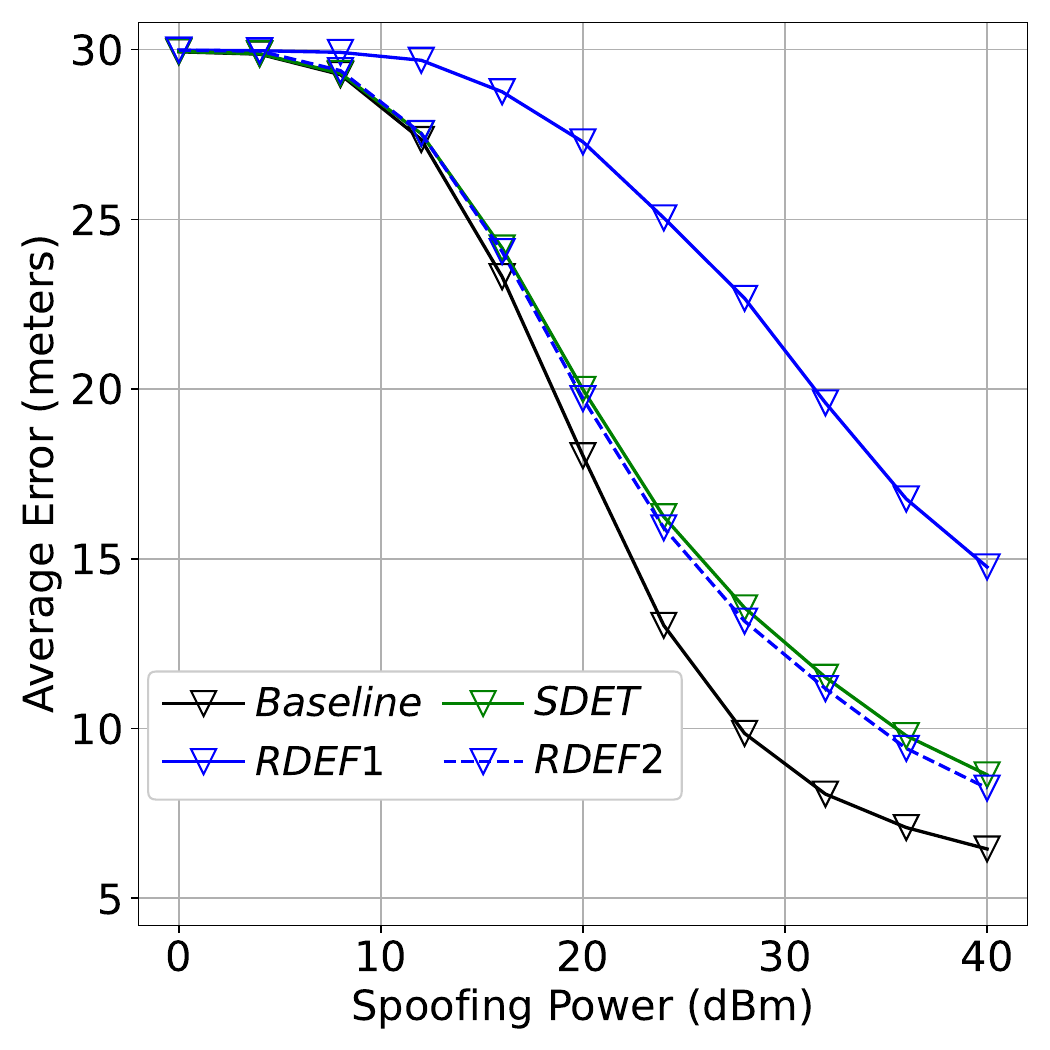}
        \caption{\emph{Global attack}}\label{fig:Glocp}
    \end{subfigure}
    \centering
   \begin{subfigure}{0.241\textwidth}
        \centering
        \includegraphics[width=\textwidth]{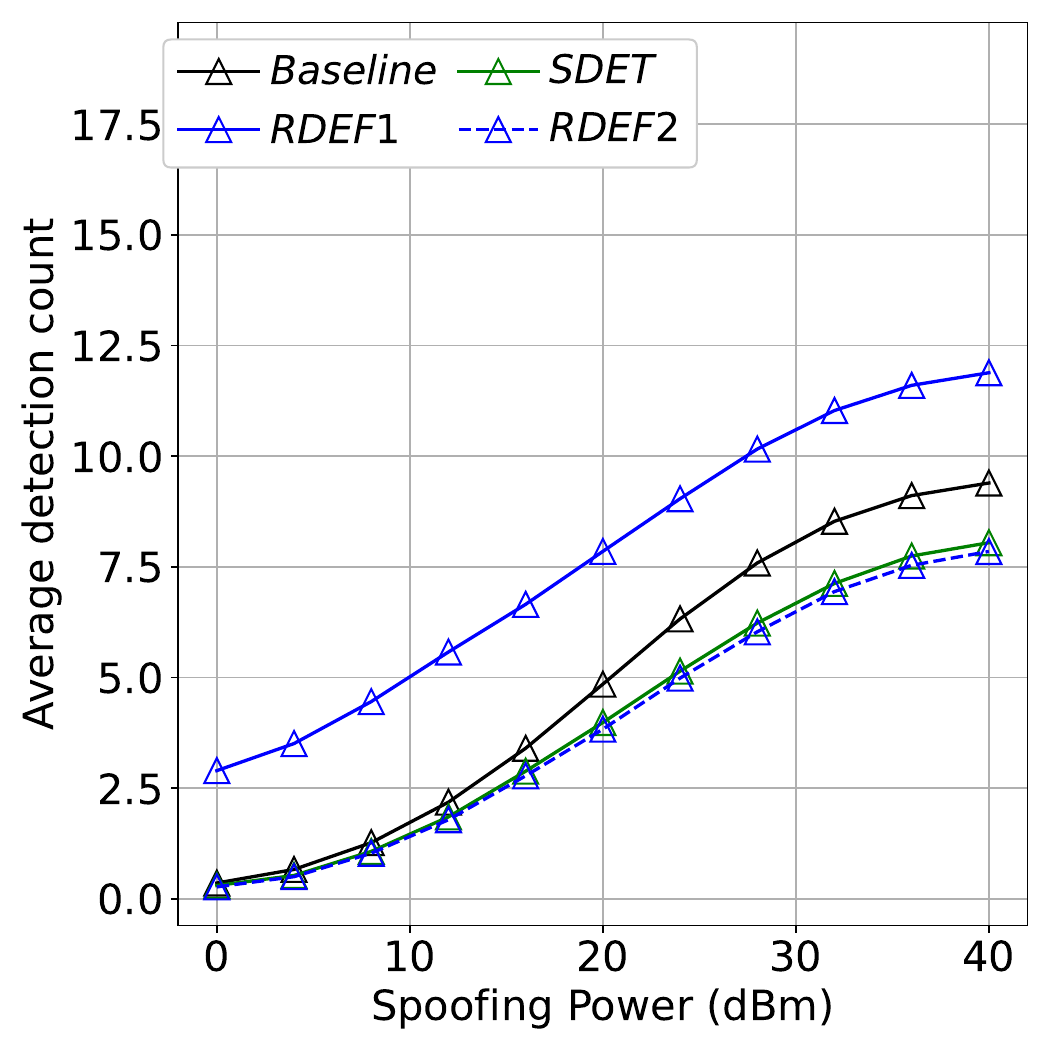}
        \caption{\emph{Focused attack}}\label{fig:Fcount}
    \end{subfigure}
   \begin{subfigure}{0.241\textwidth}
        \centering
        \includegraphics[width=\textwidth]{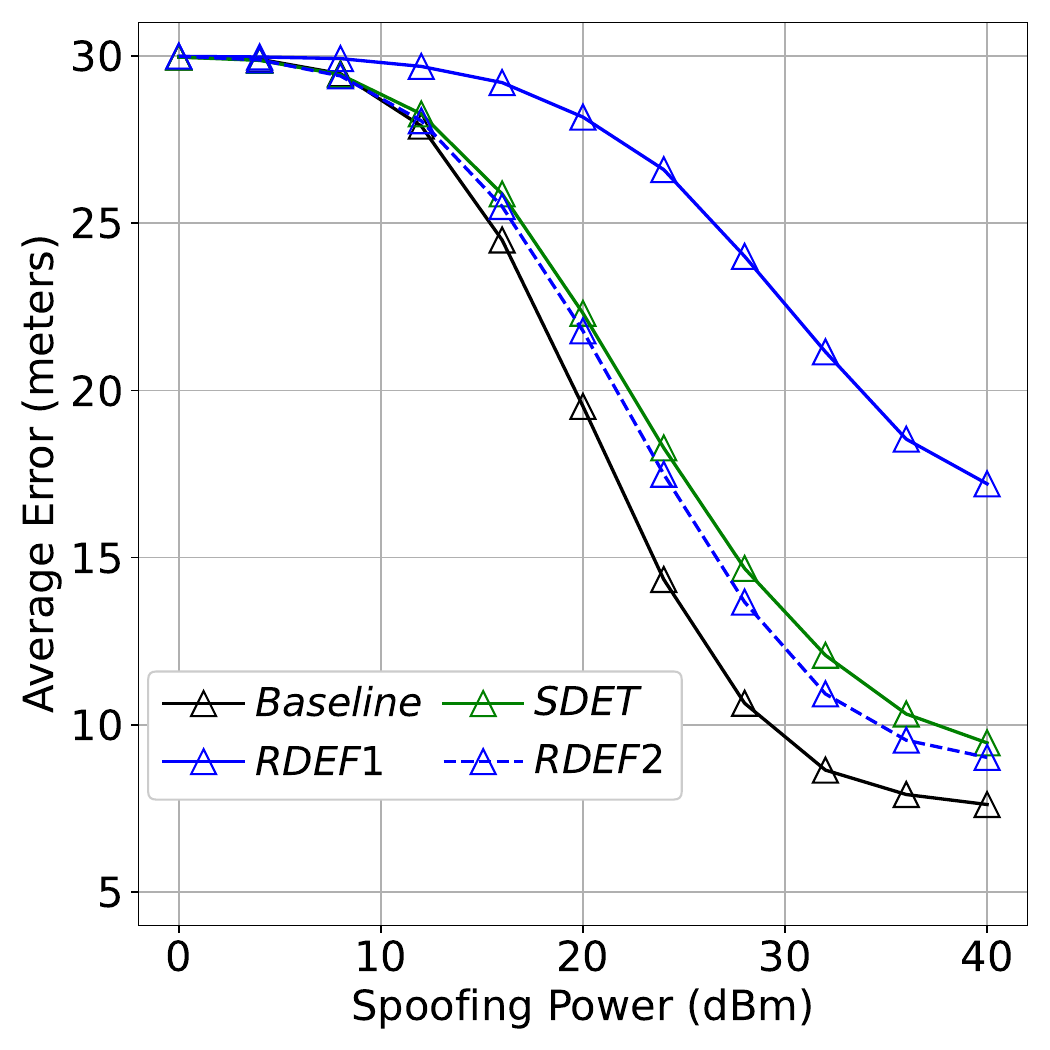}
        \caption{\emph{Focused attack}}\label{fig:Flocp}
    \end{subfigure}
    \begin{subfigure}{0.241\textwidth}
        \centering
        \includegraphics[width=\textwidth]{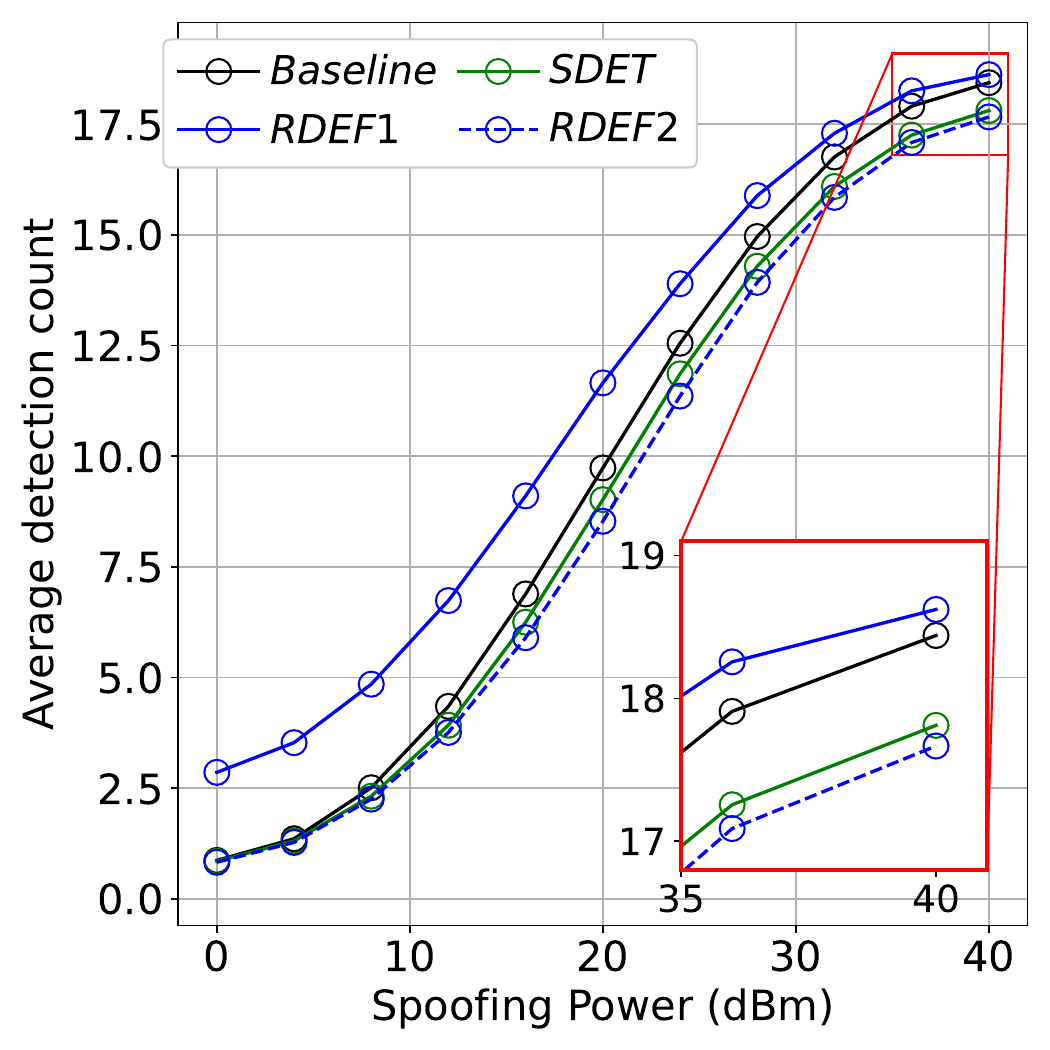}
        \caption{\emph{Selective attack} $|\bm{S}| = 5$}\label{fig:Scount}
    \end{subfigure}
   \begin{subfigure}{0.241\textwidth}
        \centering
        \includegraphics[width=\textwidth]{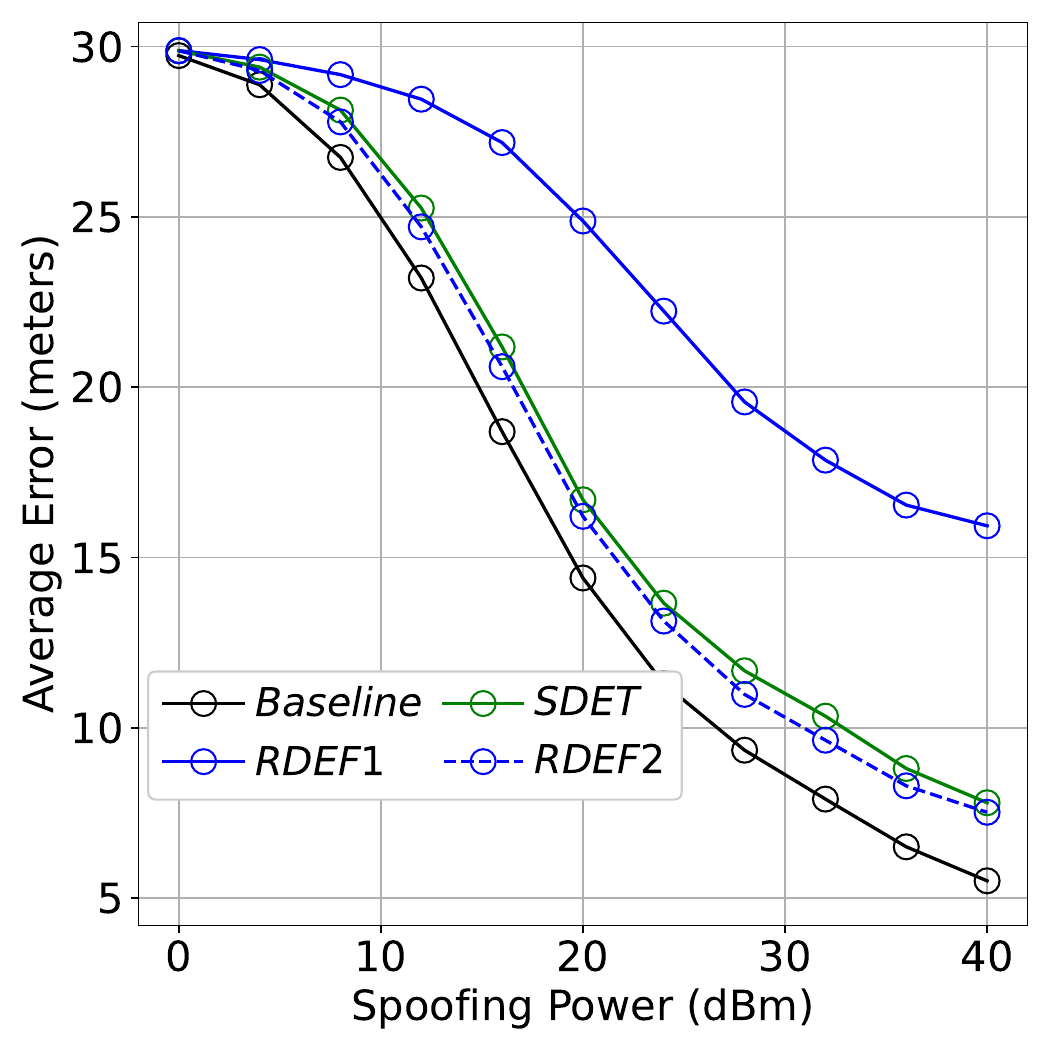}
        \caption{\emph{Selective attack} $|\bm{S}| = 5$}\label{fig:Slocp}
    \end{subfigure}
   \caption{Spoofing detection and spoofer localization error: (a), (c) and (e) show the average number of detected spoofed entries among $20$ entries over $1{,}000$ simulations with the \emph{Baseline} representing the actual spoofed entries; (b), (d) and (f) show the corresponding localization performance based on filtered anomaly data. The \emph{Baseline} is obtained by localizing the spoofer using actual spoofed entries.}
   \label{fig:resiloc}
   \end{figure}
  Second, we evaluate the spoofer localization performance, illustrating the case of a single rogue \gls{uav} attacking the localization service. For scenarios involving multiple rogue \gls{uav}s, the approach in \cite{splocfang2024} can be applied to localize multiple spoofers. Due to the high unreliability of individual $v_s^m$ estimates, spoofer localization is only performed when detected spoofed entries exceed $5$. Though \emph{RDEF1} demonstrates the best \gls{uav} localization performance in Fig.~\ref{fig:ls_Ps}. Figs.~\ref{fig:Gcount}, \ref{fig:Fcount} and \ref{fig:Scount} indicate it is over-protective with a high false-positive rate, resulting in worse spoofer localization performance in Figs.~\ref{fig:Glocp}, \ref{fig:Flocp} and \ref{fig:Slocp}. Conversely, \emph{RDEF2} and \emph{SDET} achieve better balance: despite some false negatives, their low false-positive rates yield significantly better spoofer localization than \emph{RDEF1}. In the low power regime, spoofer localization performance follows: $\emph{Selective attack} > \emph{Global attack}\approx \emph{Focused attack}$, as \emph{Selective attack} produces significantly more detected spoofed entries. In the high power regime, the ranking becomes: $\emph{Selective attack} > \emph{Global attack} > \emph{Focused attack}$. While \emph{Global attack} produces only slightly more detected entries than \emph{Selective attack}, it causes larger localization errors in $\bm{p}^m$, which propagates to spoofer localization. \emph{Focused attack} generates fewer spoofed entries but with higher power per attack, yielding less noisy $v_s^m$ estimates and maintaining comparable localization performance despite fewer entries. 
  
  Our evaluation assumes pre-optimized node selection for efficiency. The spoofer performs attacks after eavesdropping on the configured \glspl{prs}. While including additional sub-optimal nodes could enhance resilience through redundancy, this introduces significant trade-offs: increased configuration complexity, energy consumption, and localization latency. More importantly, our framework incorporates spoofer localization, enabling the system to detect and localize spoofer rather than relying solely on redundancy. 
   \section{Conclusion}\label{sec:conclu}
   This paper presents an integrated framework for performance optimization and security enhancement in \gls{3gpp}-compliant \gls{5gnr} \gls{tdoa}-based \gls{uav} localization. We demonstrate that adaptive node selection improves accuracy while reducing \gls{gnb} usage by $15\%$ compared to the statistical optimal solution based on the \gls{3gpp} \gls{a2g} channel model in small cell scenarios. We expose a novel class of merged-peak spoofing attacks that exploit signal overlap to evade existing detection methods, and through analytical modeling, quantify how synchronization quality and geometric factors determine attack success. To counter these vulnerabilities, we propose a lightweight, network-centric framework at the \gls{lmf} that integrates anomaly detection, robust victim localization, and spoofer localization using only \gls{3gpp}-specified parameters. Extensive simulations validate that our approach reduces victim localization error by $46\%$ under low-power spoofing and $71\%$ under high-power spoofing, while achieving sub-
   $10$m spoofer localization accuracy. This provides a practical foundation for secure \gls{uav} operations    in future low-altitude networks.

\bibliographystyle{IEEEtran}
\bibliography{references, IEEEabrv}
\begin{appendices}

\section{Proof of Lemma 2}\label{sec:proof3}
      \begin{figure*}[!t]
    \centering
     \begin{subfigure}[b]{0.30\textwidth}
        \centering
        \includegraphics[width=\textwidth]{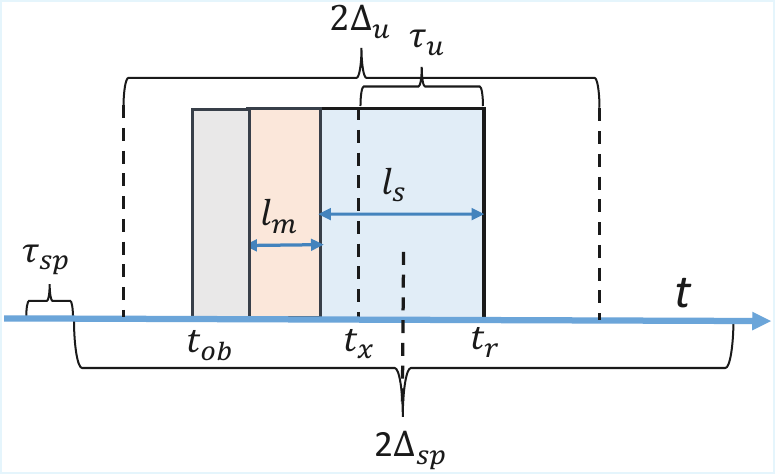}
        \caption{ Case A: $t_{ob}\leq t_r-l_s$}
    \end{subfigure}
   \begin{subfigure}[b]{0.30\textwidth}
        \centering
        \includegraphics[width=\textwidth]{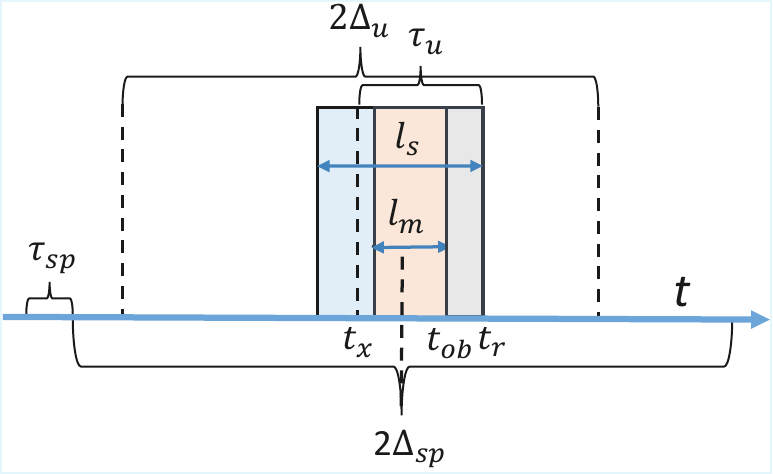}
        \caption{Case B: $t_r-l_s\leq t_{ob}\leq t_r$}
    \end{subfigure}
    \begin{subfigure}[b]{0.30\textwidth}
        \centering
        \includegraphics[width=\textwidth]{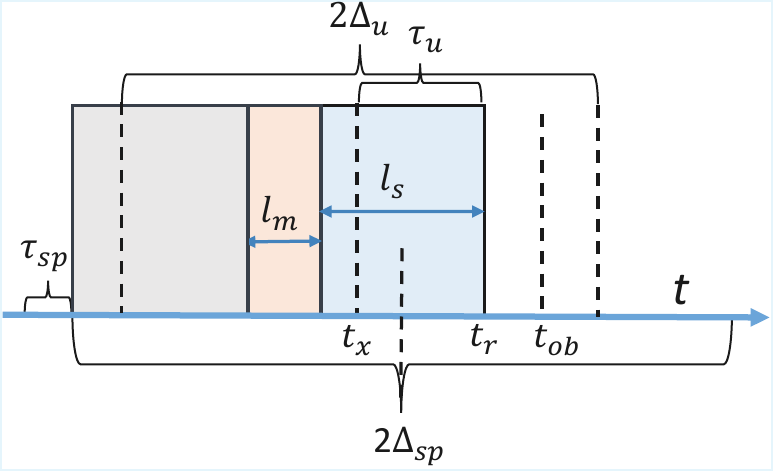}
        \caption{Case C: $t_r\leq t_{ob}$}
    \end{subfigure}
    \caption{Spoofing pulse injection: the pink and blue blocks mark the spoofed correlation peak and minimum separation. $t_{ob}$, $t_{x}$, and $t_{r}$ denote the \gls{ue} observation time, and the authentic transmission and reception times, respectively. } \label{fig:toaspoof}
    \end{figure*}
   We here consider a practical setup where, $\Delta_u > l_s+l_m$. Taking $l_s$ as the minimum peak separation requirement to resolve two distinct correlation peaks. We consider three cases and the synchronization quality is bounded by $\sigma_u\sim\mathcal{U}(0, \Delta_u)$ and $\sigma_{sp}\sim\mathcal{U}(0, \Delta_{sp})$. The probability that $t_{ob} \leq t_r$ (illustrated in Cases A and B in Fig.~\ref{fig:toaspoof}) is given by:  
  \begin{equation}\begin{split}
    P(t_{ob}\leq t_r) &= P(t_{ob}\leq t_r-l_s)+P(t_r- l_s\leq t_{ob}\leq t_r)\\
    &=\frac{\Delta_u + \tau_u-l_s}{2\Delta_u}+\frac{l_s}{2\Delta_u}=\frac{\Delta_u + \tau_u}{2\Delta_u}.
    \end{split}
   \end{equation}
   In Case A, the probability of successful spoofing is
   \begin{equation}\label{eq:pspro} 
   p_s = \min\left(\frac{l_s+l_m}{2\Delta_{sp}}, \frac{\Delta_{sp}-\tau_{sp}+\tau_u}{2\Delta_{sp}}\right).
   \end{equation}
   When the spoofer achieves very good synchronization with $\Delta_{sp} < l_s+l_m<\Delta_u$, the probability is more dominated by the second term. Otherwise, the probability is more dominated by the first term. The probability of successful spoofing with repect to two terms is \begin{equation}
   \begin{split}
    P_1(_S|A) &= \frac{(l_s+l_m)}{2\Delta_{sp}}\cdot P(t_r- l_s\leq t_{ob}\leq t_r);\quad \text{or}\\
    P_2(_S|A) &= \frac{\Delta_{sp}-\tau_{sp}+\tau_u}{2\Delta_{sp}}\cdot P(t_r- l_s\leq t_{ob}\leq t_r).\\
    \end{split}
   \end{equation}
   In Case B, similarly, we have
  \begin{equation}
   \begin{split}
    P_1(_S|B) &= \frac{(0.5l_s+l_m)}{2\Delta_{sp}}\cdot P(t_{ob}\leq t_r-l_s);\quad \text{or}\\
    P_2(_S|B) &= \frac{\Delta_{sp}-\tau_{sp}+\tau_u}{2\Delta_{sp}}\cdot P(t_{ob}\leq t_r-l_s).\\
    \end{split}
   \end{equation}
   In Case C in Fig.~\ref{fig:toaspoof} when the observation starts after the authentic puls arrives, the probability is given by: 
  \begin{equation}
    P(t_{ob}> t_r) = \frac{\Delta_u - \tau_u}{2\Delta_u}. 
   \end{equation}
   The probability of successful spoofing in Case C is:
   \begin{equation}
   \begin{split}
    P_1(_S|C) &= \frac{(l_s+l_m)}{2\Delta_{sp}}\cdot P(t_{ob}> t_r);\quad \text{or}\\
    P_2(_S|C) &= \frac{\Delta_{sp}-\tau_{sp}+\tau_u}{2\Delta_{sp}}\cdot P(t_{ob}> t_r)\leq t_r).\\
    \end{split}
   \end{equation}
   Summarizeing all cases, the overall spoofing success rate $P_S$ bounded by whichever limiting condition dominates, $P_s = \min(P_1,P_2)$. To characterize the impact of different parameters, we analyze $P_1$ and $P_2$
   separately. Taking $l_s = \alpha \times l_m$,   
   \begin{equation}
    \begin{split}
    P_1 =&   P_1(_S|A) +  P_1(_S|B) +  P_1(_S|C)\\
    =& \frac{\big(4 + 3\alpha\big)\Delta_u l_m + \alpha^2 l_m^2 - \alpha l_m\tau_u}
{8\,\Delta_{sp}\,\Delta_u},
    \end{split}
   \end{equation}
   
   \begin{equation}
   \begin{split}
    P_2 =&   P_2(_S|A) +  P_2(_S|B) +  P_2(_S|C)\\
    &= \frac{\Delta_{sp}-\tau_{sp}+\tau_u}{2\Delta_{sp}}.
   \end{split}
\end{equation}


\subsubsection{Impact of Pulse Length} 
The derivative of $P_s$ with respect to $l_m$ is given by:
\begin{equation}
\begin{split}
    \frac{\partial P_1}{\partial l_m} = \frac{(4+3\alpha)\Delta_u + 2\alpha^2l_m - \alpha\tau_u}{8\,\Delta_{sp}\,\Delta_u}\quad\text{or}\quad
    \frac{\partial P_2}{\partial l_m}= 0.
    \end{split}
\end{equation}
Since $\tau_u \ll 3\Delta_u$, we have $\frac{\partial P_s}{\partial l_s} \geq 0$, indicating that the spoofing probability $P_s$ increases with pulse length $l_s$.

\subsubsection{Impact of Authentic Signal Travel Time} 
The derivative of $P_s$ with respect to $\tau_u$ is given by:
\begin{equation}
\frac{\partial P_1}{\partial \tau_u} = -\frac{l_s}{8\,\Delta_{sp}\,\Delta_u}\quad\text{or}\quad\frac{\partial P_2}{\partial \tau_u} = \frac{1}{2\Delta_{sp}}.
\end{equation}
For $P_1$, although $\frac{\partial P_1}{\partial \tau_u} < 0$, this term dominates only when $\Delta_{sp}$ is large, making the derivative magnitude small and thus the impact negligible. In contrast, $\frac{\partial P_2}{\partial \tau_u} = \frac{1}{2\Delta_{sp}}$ becomes significant when $\Delta_{sp}$ is small. Therefore, in the regime of tight synchronization (small $\Delta_{sp}$), the authentic signal travel time $\tau_u$ significantly influences the spoofing probability $P_s$.

\subsubsection{Impact of \gls{ue} Synchronization Quality} 
The derivative of $P_s$ with respect to the synchronization quality $\Delta_u$ is given by:
\begin{equation}
\frac{\partial P_1}{\partial \Delta_u}
= \frac{\alpha\,l_m\,(\tau_u - \alpha\,l_m)}
{8\,\Delta_{sp}\,\Delta_u^{2}}\quad\text{or}\quad\frac{\partial P_2}{\partial \Delta_u}
= 0.
\end{equation}
The sign of $\frac{\partial P_s}{\partial \Delta_u}$ depends on whether $\tau_u > \alpha\,l_m$. When the \gls{gnb} is very close (i.e., $\tau_u \leq \alpha\,l_m$), we have $\frac{\partial P_s}{\partial \Delta_u} \leq 0$. In the typical case where $\tau_u > \alpha\,l_m$, we have $\frac{\partial P_s}{\partial \Delta_u} > 0$, indicating that degraded \gls{ue} synchronization facilitates spoofing. Thus, $P_s$ generally increases with $\Delta_u$, except when the authentic signal experiences minimal propagation delay.

\subsubsection{Impact of Spoofer Synchronization Error} 
The derivative of $P_s$ with respect to $\Delta_{sp}$ is:
\begin{equation}\begin{split}
\frac{\partial P_1}{\partial \Delta_{sp}}
&= -\,\frac{(4 + 3\alpha)\,\Delta_u\,l_m
      + \alpha^2 l_m^2
      - \alpha\,l_m\,\tau_u}
     {8\,\Delta_u\,\Delta_{sp}^{2}}\\
\text{or}\quad \frac{\partial P_2}{\partial \Delta_{sp}}
&= \frac{\tau_{sp} - \tau_u}{2\,\Delta_{sp}^{2}}.
\end{split}
\end{equation}
Since $\tau_u \ll 4\Delta_u+3\alpha\Delta_u+\alpha l_m$, we have $\frac{\partial P_1}{\partial\Delta_{sp}} < 0$, indicating that $P_s$ decreases as spoofer synchronization degrades. For $P_2$, the sign of $\frac{\partial P_2}{\partial \Delta_{sp}}$ depends on the relative geometry. In the typical scenario where the spoofer is closer to the \gls{ue} than the legitimate \gls{gnb} (i.e., $\tau_{sp} < \tau_u$), we have $\frac{\partial P_2}{\partial \Delta_{sp}} < 0$, meaning improved spoofer synchronization significantly enhances attack success. Only when $\tau_{sp} > \tau_u$ does this relationship weaken or reverse.
\end{appendices}

\end{document}